\documentclass[11pt]{amsart}
\usepackage[utf8]{inputenc}
\usepackage{amsmath,amssymb,amsthm,hyperref}
\usepackage{url}
\usepackage[left=1in,right=1in,top=1in,bottom=1in]{geometry}
\usepackage{graphicx}
\usepackage[square,sort,comma,numbers]{natbib}
\usepackage[dvipsnames,usenames]{xcolor} 
\usepackage{mathtools}
\usepackage{cleveref,enumitem}
\usepackage{tikz}
\usepackage{comment}
\usepackage{multirow}
\usepackage{arydshln}
\usepackage[linesnumbered,ruled,vlined]{algorithm2e}

\usepackage{nicematrix}
\usepackage{ulem}
\usepackage{lineno}

\makeatletter
\newcommand{\thickhline}{%
    \noalign {\ifnum 0=`}\fi \hrule height 1pt
    \futurelet \reserved@a \@xhline
}
\newcolumntype{"}{@{\hskip\tabcolsep\vrule width 1pt\hskip\tabcolsep}}
\makeatother

\newtheorem{theorem}{Theorem}
 \numberwithin{theorem}{section}
\newtheorem{proposition}[theorem]{Proposition}
\newtheorem{lemma}[theorem]{Lemma}
\newtheorem{corollary}[theorem]{Corollary}

\theoremstyle{definition}

\newtheorem{remark}[theorem]{Remark}
\newtheorem{example}[theorem]{Example}

\graphicspath{{Figures/}}

\newcommand{\RR}{\mathbb{R}}

\newcommand{\W}{\det\big( [W_n]_{[3],I} \big)}
\renewcommand{\k}{\kappa}
\renewcommand{\l}{\lambda}

\newcommand{\blockW}[1]{
	\underbrace{\begin{matrix}\overline{W} & \cdots & \overline{W}\end{matrix}}_{#1}
}

\DeclareMathOperator{\im}{im}
\DeclareMathOperator{\rk}{rk}
\DeclareMathOperator{\diag}{diag}
\DeclareMathOperator{\sgn}{sgn}
\DeclareMathOperator{\supp}{supp}
\DeclareMathOperator{\N}{NP}
\DeclareMathOperator{\Conv}{Conv}
\DeclareMathOperator{\Id}{Id}
\DeclareMathOperator{\Vertex}{Vert}
\DeclareMathOperator{\pr}{pr}
\DeclareMathOperator{\rank}{rank}

\title[Connectivity of Parameter Regions for Multistationarity]{Connectivity of Parameter Regions of Multistationarity for Multisite Phosphorylation Networks}
\author{Nidhi Kaihnsa and Máté L. Telek }

\begin{document}

\maketitle

\begin{abstract}
The parameter region of multistationarity of a reaction network contains all the parameters  for which the associated dynamical system exhibits multiple steady states. Describing this region is challenging and remains an active area of research. In this paper, we concentrate on two biologically relevant families of reaction networks that model multisite phosphorylation and dephosphorylation of a substrate at $n$ sites. For small values of $n$, it had previously been shown that the parameter region of multistationarity is connected. Here, we extend these results and provide a proof that applies to all values of $n$. Our techniques are based on the study of the critical polynomial associated with these reaction networks together with polyhedral geometric conditions of the signed support of this polynomial.

    \vskip 0.1in
	

	\noindent {\bf 2020 MSC:}{ 92xx, 52Bxx 
		
	}
\end{abstract}


\section{\bf Introduction}
 Within the framework of reaction network theory \cite{feinbergbook}, the change in concentration of the species over time is modeled by a parametrized ordinary differential equation (ODE) system. In this paper, we focus on a fundamental property of ODE systems: the existence of multiple steady states, also known as \textit{multistationarity}. Having multistationarity is a precursor to {\em multistability}, which has been linked to cellular decision-making, switching, and the memory of cells \cite{CellSignaling, SwitchLike}.

Under the assumption of mass-action kinetics, the functions in the ODE system become polynomials, which are parametrized by two types of parameters: reaction rate constants and total concentrations. Identifying the \emph{parameter region of multistationarity} is equivalent to describing the set of parameters for which the polynomial equation system has at least two positive real solutions. While existing symbolic methods, such as Cylindrical Algebraic Decomposition and Quantifier Elimination, may offer a complete description of this region, their high algorithmic complexity limits their applicability to reaction networks of moderate size \cite{BRADFORD202084}. On the other hand, numerical methods are able to handle larger reaction networks and offer insights into specific parts of the parameter space, but they do not provide information about the entire parameter space \cite{Sadeghimanesh:KacRice,HarringtonHauenstein,ParamGeo}.

{
	Even though an exact description of the multistationarity region is usually out of reach, having some partial information about its properties, such as its shape, could have significant biological implications \cite{CSS}. For instance, shape of the parameter region of multistationarity has been associated with robustness. Two different connected regions with same volume but different shapes will not be identically robust to the parametric perturbations for multistationarity. 
	
	Determining connectivity is one of the quantifying descriptors of the shape of the desired parameter region. Furthermore, multistability in reaction networks give rise to bistable switches relevant for different biological processes. {One of the common switching behavior is hysteresis \cite{HYMWQ}.} Having more than one connected components for multistability allows for the possibility of more complicated switching behaviours \cite{ParamGeo}  {since going continuously from one connected region to another requires passing through the monostationarity region.} Owing to this understanding, shapes of parameter regions for multistationarity and its connectivity has garnered much attention in recent years \cite{dennis2023connectedness,MultDualPhos,Multnsite,ConnectivityPaper}. In this paper, we continue this line of research and investigate the parameter region of multistationarity for two infinite families of reaction networks modeling phosphorylation mechanisms.}

{ The networks considered in this article are well-studied examples of protein post-translational modification (PTM) mechanism which plays a crucial role in cell signaling processes \cite{G-distributivity,rendall-MAPK}. PTM mechanisms involving phosphorylation process have been found relevant for modification of about 30\% of all proteins in human body \cite{Cohen}. There are proteins that can have large number of modification sites \cite{GRC} which can make modeling their biological behavior a combinatorially complex process \cite{BMHK, Bris}. One of the widely studied proteins for different post-translational modifications is p53 which is considered as tumor supressor protein. For this protein, phosphorylation process, in particular, is known to act as a crucial regulator \cite{KG,GZ}. }

Phosphorylation mechanisms typically involve a substrate denoted by $S$, along with a kinase $K$ and a phosphatase $F$, responsible for catalyzing the phosphorylation and dephosphorylation of $S$, respectively. Assuming that both phosphorylation and dephosphorylation occur in a sequential and distributed manner, and the substrate $S$ has $n \in \mathbb{N}$ phosphorylation sites, the corresponding reaction network takes the following form:
{\small
\begin{align}\tag{$\mathcal{N}_{1,n}$}
\begin{split}
\label{Eq::nsitenetwork}
S_{i} +K \xrightleftharpoons[\kappa_{6i+2}]{\kappa_{6i+1}} &Y_{i+1} \xrightarrow{\kappa_{6i+3}} S_{i+1} + K  \\ \notag
S_{i+1} +F \xrightleftharpoons[\kappa_{6i+5}]{\kappa_{6i+4}} &U_{i+1} \xrightarrow{\kappa_{6i+6}} S_{i} + F, 
\end{split}
\qquad \quad i =0,\dots,n-1.
\end{align}
}
We refer to network \eqref{Eq::nsitenetwork} as the \emph{strongly irreversible $n$-site phosphorylation network}.

In \cite{MultDualPhos}, it was shown that the projection of the parameter region of multistationarity onto the reaction rate constants is connected for $n=2$. This result was later generalized for all $n \geq 2$ in~\cite{Multnsite}. Subsequently, it was shown with the aid of a computer that the entire multistationarity region is connected for $n = 2,3$ in \cite{ConnectivityPaper}, and for $n = 4,5,6,7$ in \cite{FaceReduction}. One of the main contribution of the current article is a proof showing that for  \eqref{Eq::nsitenetwork}, the parameter region of multistationarity is connected for every $n \geq 2$.

The second family of phosphorylation networks considered in this paper is the \emph{weakly irreversible $n$-site phosphorylation network}, which is 
{\small
\begin{align}
\label{Eq::Gunawardena}\tag{$\mathcal{N}_{2,n}$}
\begin{split}
S_{i} +K \xrightleftharpoons[\kappa_{10i+2}]{\kappa_{10i+1}} &Y_{2i+1}  \xrightarrow{\kappa_{10i+3}}   Y_{2i+2}  \xrightleftharpoons[\kappa_{10i+5}]{\kappa_{10i+4}}  S_{i+1} + K\\ \nonumber
S_{i+1} +F \xrightleftharpoons[\kappa_{10i+7}]{\kappa_{10i+6}} &U_{2i+1} \xrightarrow{\kappa_{10i+8}} U_{2i+2} \xrightleftharpoons[\kappa_{10i+10}]{\kappa_{10i+9}}S_{i} + F, 
\end{split}
\qquad i =0,\dots,n-1.
\end{align}
}
The name of these networks originates from the fact that once a phosphate group is attached, the product $S_{i+1} + K$ might rebind to form the intermediate species $Y_{2i+2}$. Analogously, for dephosphorylation, the product $S_i + F$ might rebind to $U_{2i+2}$. In \cite{ParamGeo}, it has been argued that allowing this product rebinding for phosphorylation systems is more realistic compared to the mechanism represented in \eqref{Eq::nsitenetwork}.
For \eqref{Eq::Gunawardena} with $n=2$, the authors in \cite{ParamGeo} investigated the shape of the parameter region of multistationarity numerically, and their methods strongly indicated that this region is connected. In \cite{FaceReduction}, a symbolic proof was provided showing that the multistatinarity region of ($\mathcal{N}_{2,2}$) is indeed connected.  

 In this paper we prove that for both families \eqref{Eq::nsitenetwork} and \eqref{Eq::Gunawardena} the parameter region of multistationarity is connected for every $n \geq 2$. Our approach builds upon \cite{PLOS_IdParaRegions}, where the authors associated a polynomial, called the {\em critical polynomial}, with reaction networks satisfying some mild assumptions. Subsequently, in \cite[Theorem 4]{ConnectivityPaper}, the authors gave conditions on the critical polynomial that imply the connectivity of the parameter region of multistationarity.

One of the main challenges faced is that a direct computation of the critical polynomial for the networks \eqref{Eq::nsitenetwork} and \eqref{Eq::Gunawardena} becomes infeasible even for relatively small $n$. In \cite{ConnectivityPaper}, it was shown that it is enough to compute the critical polynomial of a reduced version of the reaction network. In~\cite{FaceReduction}, this reduced critical polynomial was successfully computed for \eqref{Eq::nsitenetwork} when $n = 2,3,4,5,6,7$ and for \eqref{Eq::Gunawardena} when $n = 2$. Even for the reduced critical polynomial, the computation becomes intractable for larger values of $n$. To overcome this difficulty, in Section~\ref{Sec::CritandGale}, we derive a formula for the critical polynomial (Theorem \ref{Prop:CritPolyWnDn}) using Gale dual matrices that allows us to work with the critical polynomial for all $n$. We believe this result is interesting in its own right, as it might simplify the computation of the critical polynomial for general networks. 
 
Section~\ref{Section::SignedSupp} and Section~\ref{Sec::EdgesNP}, discuss the structure of the Newton polytope and the signed support of a polynomial that we later employ to study the critical polynomial. We explain our approach to prove connectivity of the parameter region of multistationarity in Section \ref{Section:Overview}. In Sections~\ref{Section::SIP} and \ref{Section::Weakly} we prove connectivity for \eqref{Eq::nsitenetwork}  and \eqref{Eq::Gunawardena} respectively for every $n \geq 2$.

\section{\bf Preliminaries}
\label{Sec::Prelim}

To prove that the parameter region of multistationarity is connected for both families of reaction networks \eqref{Eq::nsitenetwork} and \eqref{Eq::Gunawardena} for every  $n \geq 2$, in Sections~\ref{Section::SIP} and \ref{Section::Weakly}, we apply several interconnected mathematical results. To provide the reader with an overview of the central mathematical objects and to fix the notation, we summarize these objects in Table~\ref{Table}.

\begin{table}[!t]

\centering
\caption{
{\bf Table of notation}}
\begin{tabular}{c|l}
\hline
$\qquad n \qquad $ & \small {number of sites} \\ \hline
$\qquad m \qquad $ & \small number of species \\ \hline
$\qquad r \qquad $ & \small number of reactions \\ \hline
$\qquad N {\in \mathbb{Z}^{m \times r}} \qquad $ & \small stoichiometric matrix \\ \hline
$\qquad A {\in \mathbb{Z}^{m \times r}} \qquad $ & \small reactant matrix \\ \hline
$\qquad \kappa  {\in \mathbb{R}^{r}_{>0}} \qquad $ & \small reaction rate vector \\ \hline
$\qquad s {= \rk(N)} \qquad $ & \small {rank of the stoichiometric matrix} \\ \hline
$\qquad d { = m-s}\qquad $ & \small {number of conservation laws} \\ \hline
$\qquad T { \in \mathbb{R}_{>0}^d}\qquad $ & \small total concentration vector \\ \hline
$\qquad \mathcal{P}_T {\subseteq \mathbb{R}^m_{\geq0}}\qquad $ & \small stoichiometric compatibility class \\ \hline
$\qquad W { \in \mathbb{R}^{d \times m}}\qquad $ & \small conservation law matrix \\ \hline
$\qquad V_\kappa { \subseteq \mathbb{R}^m_{\geq0}}\qquad$ & \small steady state variety \\ \hline
$\qquad \Omega { \subseteq \mathbb{R}^r_{>0} \times \mathbb{R}^d } \qquad $ & \small parameter region of multistationarity \\ \hline
$\qquad \mathcal{C} {=\ker(N) \cap \mathbb{R}^r_{\geq0}} \qquad $ & \small flux cone \\ \hline
$\qquad E {\in \mathbb{R}^{r \times \ell}}\qquad $ & \small extreme matrix \\ \hline
$\qquad E^{(1)}, \dots,E^{(\ell)} {\in \mathbb{R}^{\ell}}  \qquad $ & \small extreme vectors \\ \hline
$\qquad q(h,\lambda) \qquad $ & \small critical polynomial \\ \hline
$\qquad h_1, \dots h_m, \quad \lambda_1, \dots, \lambda_\ell \qquad $ & \small variables of the critical polynomial \\ \hline
$\qquad D(\lambda) {\in \mathbb{R}(\lambda)^{m \times d}}\qquad $ & \small Gale dual matrix of $N'\diag(E\lambda)A^\top$ \\ \hline
\end{tabular}
\label{Table}
\end{table}

The goal of the following subsections is to elaborate on these objects and prove additional results that we will use in Sections~\ref{Section::SIP} and \ref{Section::Weakly}. The key player of the proofs in Sections~\ref{Section::SIP} and \ref{Section::Weakly} is the \emph{critical polynomial}, which we will discuss in Section~\ref{Sec::ConvexParameters} along with the \emph{flux cone}. It is known that if the critical polynomial satisfies certain connectivity and closure properties (Theorem~\ref{Thm::ReducedConnectivity}), then the parameter region of multistationarity is connected. To verify these properties, we rewrite the critical polynomial in terms of a \emph{Gale dual matrix} and use combinatorial properties of its \emph{signed support}. A detailed overview of the method is provided in Section~\ref{Section:Overview}.

\color{black}

\subsection{Reaction networks and multistationarity}\label{sec:rxnnetworks}
A \emph{reaction network} over species $X_1,\ldots,X_m$ is a collection of \emph{reactions} of the form 
$\sum_{i=1}^m a_{ij} X_i \to \sum_{i=1}^m b_{ij} X_i$ for $j = 1, \dots , r$,
where $a_{ij}, b_{ij}$ are non-negative integers. Each reaction is weighted by a positive parameter $\kappa_j \in \RR_{>0}$, called \textit{reaction rate constant}. By $\k:=(\k_1,\ldots,\k_r)\in\RR_{>0}^r$, we denote the reaction rate vector.

The net production of the species along each reaction is encoded in the \emph{stoichiometric matrix} $N\in \mathbb{Z}^{m\times r}$, which is defined as
\[ N := \big(b_{ij} - a_{ij}\big)  \in \mathbb{Z}^{m \times r}.\]
We consider also the \textit{reactant matrix}, given by
\[ A :=  \big( a_{ij} \big) \in \mathbb{Z}^{m \times r}.\]
Using these matrices and under the assumption of \emph{mass-action kinetics}, the ODE system that models the evolution of the concentration of the species is
\begin{align}
\label{Eq::ODE}
\dot{x} =  N \diag(\kappa) x^{A},
\end{align}
where $x =  (x_1, \dots , x_m)\in \RR^{m}_{\geq 0}$ is a vector representing the concentration of the species $X_1, \dots , X_m$, and $x^A:=\big(\prod_{i=1}^{m}x_i^{a_{i1}},\ldots,\prod_{i=1}^{m}x_i^{a_{ir}}\big)^{\top}\in\RR_{\geq0}^{r}$.

For the ODE system in \eqref{Eq::ODE}, the trajectories are contained in affine linear subspaces of $\RR^m$ called \emph{stoichiometric compatibility classes}. These are affine translates of the image of the stoichiometric matrix $N$. The dimension of this subspace will be denoted by $ s = \rk(N)$. We represent $\im(N)$ by linear equations determined by a full-rank matrix $W \in \mathbb{R}^{d \times m}$ such that $W N = 0$ and $d = m - s$. Such a matrix is called a \emph{conservation law matrix}. Using $W$, we define the stoichiometric compatibility classes as
\[ \mathcal{P}_T := \{ x \in \mathbb{R}^{m}_{\geq 0} \mid W x = T \},\]
where $T \in \mathbb{R}^d$ is called the \emph{total concentration vector}. A reaction network is \emph{conservative}, if every stoichiometric compatibility class is a compact set. Equivalently, there exists a vector with only positive coordinates in the left kernel of $N$ (cf. \cite{BENISRAEL1964303}).

\begin{example}
To illustrate the above definitions, we consider the following small reaction network, which we will refer to as the \emph{running example}.
    \begin{align*}&4 X_1 \xrightarrow{\kappa_{1}}  2 X_1 + X_2+X_3, \quad
    &6 X_2 \xrightarrow{\kappa_{2}}  X_1 + 2 X_2+ 3 X_3, \\
    &5 X_3 \xrightarrow{\kappa_{3}}  X_1 + 3 X_2 +  X_3, \quad
    &4 X_1 + 2 X_2 + 6 X_3 \xrightarrow{\kappa_{4}}  2  X_1 + 3 X_2 + 7 X_3.
     \end{align*}
This reaction network has $3$ species $X_1,\, X_2,\,  X_3$ and $4$ reactions. The stoichiometric and reactant matrix are given by
    \begin{align*}
        N = \begin{pNiceArray}{cccc}
         -2&\phantom{-}1&\phantom{-}1&-2\\
         \phantom{-}1&-4&\phantom{-}3&1\\
         \phantom{-}1&\phantom{-}3&-4&\phantom{-}1\\
     \end{pNiceArray} \in \mathbb{Z}^{3\times 4}, \qquad    A = \begin{pNiceArray}{cccc}
         4&\phantom{-}0&\phantom{-}0&\phantom{-}4\\
         0&\phantom{-}6&\phantom{-}0&\phantom{-}2\\
         0&\phantom{-}0&\phantom{-}5&\phantom{-}6\\
     \end{pNiceArray} \in \mathbb{Z}^{3\times 4},
    \end{align*}
and the ODE system~\eqref{Eq::ODE} consists of the following three polynomial equations:
        \begin{align*}
      \dot{x}_1 &= -2 \kappa_1 x_1^4 + \phantom{1}\kappa_2 x_2^6 + \phantom{1} \kappa_3 x_3^5 - 2 \kappa_4 x_1^4x_2^2x_3^6,  \\ 
         \dot{x}_2 &=  \phantom{-}\phantom{1}\kappa_1 x_1^4 -4 \kappa_2 x_2^6 + 3 \kappa_3 x_3^5 + \phantom{1}\kappa_4 x_1^4x_2^2x_3^6,  \\ 
            \dot{x}_3 &=  \phantom{-}\phantom{1}\kappa_1 x_1^4 + 3\kappa_2 x_2^6 -4 \kappa_3 x_3^5 + \phantom{1} \kappa_4 x_1^4x_2^2x_3^6.  
    \end{align*}
By computing the left-kernel of $N$, one can find that the conservation law matrix can be chosen as
\begin{align*}
      W =  \begin{pNiceArray}{ccc}
       1& 1 & 1\\
     \end{pNiceArray} \in \RR^{1\times 3}.
    \end{align*}
Since $W$ has only positive coordinates, it follows that the running example is conservative.
\end{example}

\color{black}
  
Given a reaction network, the set of \textit{steady states} is obtained by the common zeros of the polynomials on the right-hand side of \eqref{Eq::ODE}. As we are only interested in the non-negative steady states, for fixed reaction rate constants $\kappa$, we consider the \emph{steady state variety}
\begin{align}\label{eqn:steadystatevariety} V_{\kappa} := \{ x \in \mathbb{R}^{m}_{\geq 0} \mid  N \diag(\kappa) x^{A} = 0\}.\end{align}
A steady state $x \in V_{\kappa}$ is a \emph{relevant boundary steady state} if one of the coordinates of $x$ equals zero and the stoichiometric compatibility class containing $x$ intersects the positive orthant $\mathbb{R}^{m}_{>0}$ . 

A parameter pair $(\kappa,T) \in \mathbb{R}^{r}_{>0} \times \mathbb{R}^d$ \emph{enables multistationarity} if $V_{\kappa} \cap \mathcal{P}_T \cap \mathbb{R}_{>0}^{m}$ contains at least two points. The \emph{parameter region of multistationarity} is thus defined as
\begin{align} \label{Eq::Omega} \Omega := \{(\kappa,T) \in \mathbb{R}^{r}_{>0} \times \mathbb{R}^d \mid ( \kappa,T) \text{ enables multistationarity} \}.\end{align}

 In this article, our main goal is to show that the set $\Omega$ is path connected for the two families of networks \eqref{Eq::nsitenetwork} and \eqref{Eq::Gunawardena}. 

\subsection{Critical polynomial}\label{Sec::ConvexParameters}
In \cite{ConnectivityPaper}, the authors described a sufficient condition based on a polynomial that implies connectivity of the parameter region of multistationarity. In this section, we recall this statement (Theorem \ref{Thm::ReducedConnectivity}) and elaborate on how to compute this polynomial.

Given a reaction network with stoichiometric matrix $N$, the set $\mathcal{C} := \ker(N) \cap \mathbb{R}^{r}_{\geq0}$ is a convex polyhedral cone, called the \emph{flux cone} \cite{SNA-Book}. {The importance of this cone lies in the simple observation that $\diag(\kappa)x^A \in \mathcal{C}$ whenever $x$ lies in the steady state variety given by $\kappa$. To describe the flux cone, we consider a minimal collection of its generators $\{ E^{(1)}, \dots , E^{(\ell)} \} \subseteq \mathbb{R}^{r}$. These generators are called a choice of \emph{extreme vectors}. } The extreme vectors of $\mathcal{C}$ are unique up to multiplication by positive scalars \cite{Rockafellar}. A matrix $E\in \RR^{r\times \ell}$, whose columns are given by a choice of extreme vectors, is called an \textit{extreme matrix}. If $E$ does not have a zero row, the reaction network is called \textit{consistent}. This property is equivalent to $\ker(N) \cap \mathbb{R}^r_{>0} \neq \emptyset$.

In Proposition~\ref{Prop::ExtremeVector}, we establish a condition that ensures that a basis of $\ker(N)$ gives a choice of extreme vectors of $\mathcal{C}$. We denote by $[r]$ the set $\{1, \dots, r \}$, and for a vector $u \in \mathbb{R}^r$, we write
\[ \supp(u) := \big\{ i \in [r] \mid u_i \neq 0 \big\}.\] 
A vector $v \in \mathcal{C} \setminus \{0\}$ is an extreme vector if and only if $v$ is \emph{support-minimal} \cite{MuellerReg_ExtremeVectors}, i.e. for all non-zero $w \in \mathcal{C}$ it holds
\begin{align}
\label{Eq::SuppMinimal}
\supp(w) \subseteq \supp(v) \quad \text{implies} \quad \supp(w) = \supp(v).
\end{align}
Using this we now prove the following result.

\begin{proposition}
\label{Prop::ExtremeVector}
Consider a reaction network with stoichiometric matrix $N \in \mathbb{Z}^{m \times r}$. Assume that $E^{(1)}, \dots , E^{(\ell)} \in \mathbb{R}^r_{\geq0}$ is a basis of $\ker(N)$, and for every $k \in [\ell]$ there exists $j_k \in [r]$  such that
\begin{align}
\label{Eq::Assumption::ExtremeVector}
 j_k \in \supp(E^{(k)}) \setminus \big( \bigcup_{\substack{i=1\\ i \neq k}}^{\ell} \supp(E^{(i)}) \big). 
 \end{align}
Then $\{E^{(1)}, \dots , E^{(\ell)}\}$ is a choice of extreme vectors of the flux cone $\mathcal{C}$.
\end{proposition}

\begin{proof}
    By \eqref{Eq::SuppMinimal}, it is enough to show that $E^{(1)}, \dots , E^{(\ell)}$ are the only support-minimal vectors in $\mathcal{C}$ up to multiplication by a scalar. Consider $w \in \mathcal{C}\subset \ker(N)$ with $w \neq 0$. By assumption, there exist $a_1, \dots ,a_\ell \in \mathbb{R}$ such that
    $w = \sum_{k=1}^{\ell} a_k E^{(k)}$. As $j_k$ satisfies \eqref{Eq::Assumption::ExtremeVector}, $w_{j_k} = a_k E^{(k)}_{j_k}$ for $k \in [\ell]$. Since $E^{(k)}_{j_k} > 0$, we conclude that $a_1, \dots , a_\ell \geq 0$ and hence,
    \begin{align}
    \label{Eq::Proof::ExtremeVector}
    \supp(E^{(k)}) \subseteq \supp(w) \quad \text{if } a_k \neq 0.
     \end{align}

   If $\supp(w) \subseteq \supp(E^{(k)})$ for some $k$, then $\supp(w) = \supp(E^{(k)})$. Consequently, $E^{(k)}$ is support-minimal, and hence $E^{(k)}$ is an extreme vector for all $k \in  [\ell]$.
    
    If $w$ is an extreme vector of $\mathcal{C}$, then there exists $k\in [\ell]$ with $a_k \neq 0$. From \eqref{Eq::Proof::ExtremeVector} and the support-minimality of $w$, it follows that $\supp(E^{(k)}) = \supp(w)$. If there exist two distinct $k_1, k_2 \in[\ell]$ such that $a_{k_1},a_{k_2}$ are non-zero, then 
    $\supp(E^{(k_1)}) = \supp(w) = \supp(E^{(k_2)})$, which contradicts \eqref{Eq::Assumption::ExtremeVector}. Thus, there exists exactly one non-zero $a_k$ and $w = a_k E^{(k)}$.
\end{proof}

\begin{example}
\label{Ex:basis}
    For the running example, a simple computation shows that the vectors
\begin{align*}
    E^{(1)} = \begin{pNiceArray}{c}
         1 \\ 1 \\ 1\\ 0
     \end{pNiceArray} \qquad E^{(2)} = \begin{pNiceArray}{c}
         0 \\ 1 \\ 1\\ 1
     \end{pNiceArray} 
\end{align*}
lie in the kernel of $N$. The supports of these vectors equal 
\[\supp(E^{(1)}) = \{1,2,3\}, \quad \supp(E^{(2)}) = \{2,3,4\}.\]

Since the vectors $E^{(1)}, E^{(2)}$ form a basis of $\ker(N)$, $ \supp(E^{(1)}) \setminus \supp(E^{(2)}) = \{1\}$ and $ \supp(E^{(2)}) \setminus \supp(E^{(1)}) = \{4\}$, Propositon~\ref{Prop::ExtremeVector} implies that $E^{(1)}, E^{(2)}$ is a choice of extreme vectors of the flux cone $\mathcal{C} = \ker(N) \cap \mathbb{R}^{4}_{\geq0}$.
\end{example}
\color{black}

{Using extreme vectors, one can reparametrize Jacobian matrices of the function given by the right-hand side of~\eqref{Eq::ODE} evaluated at steady states \cite{JacParam,ConnectivityPaper}. Specifically, each Jacobian can be written as
\begin{align}
\label{Eq:JacParam}
    N \diag(E\lambda) A^\top \diag(h),
\end{align}
for some $h=(h_1,\ldots,h_m)\in \RR^{m}_{>0}, \, \lambda=(\lambda_1,\ldots,\lambda_{\ell})\in \RR^{\ell}_{>0}$. To define the critical polynomial, we modify the matrix~\eqref{Eq:JacParam} as follows.} We assume that $W$ is row reduced and that $1,\dots,d$ are the indices of the first non-zero entries of the rows of $W$. In Sections~\ref{Section::SIP} and \ref{Section::Weakly}, we will choose the conservation law matrix for various networks such that this extra assumption is satisfied (cf.  \eqref{Eq::Wn_nsite}, \eqref{Eq::Wn_Gunawardena}). We define 
\begin{align}
\label{Eq::Mntilde}
M(h,\lambda) := \begin{pmatrix} W \\ N' \diag(E \lambda )A^{\top} \diag(h) \end{pmatrix} \in \mathbb{R}^{m\times m},
\end{align}
where $N' \in \mathbb{R}^{s \times r} $ is the matrix obtained from $N$ by deleting the first $d$ rows. Following \cite{MultStrucRN,ConnectivityPaper}, we consider the following function
\begin{align}
\label{Eq:CritPoly}
q\colon \mathbb{R}^{m}_{>0} \times \mathbb{R}^{\ell}_{>0} \to \mathbb{R}, \quad (h,\lambda) \mapsto q(h,\lambda) := (-1)^{s}\det M(h,\lambda),
\end{align}
and call $q(h,\l)$ the \emph{critical polynomial}.

\begin{example}
\label{Ex:Mhl}
For the running example, we have
    \begin{align*}
     M(h,\lambda) =  \begin{pNiceArray}{ccc}
         1 & 1 & 1 \\
          4\lambda_1h_1 + 4\lambda_2h_1 &  -24\lambda_1h_2 - 22\lambda_2h_2 &  15\lambda_1h_3 + 21\lambda_2h_3 \\
           4\lambda_1h_1 + 4\lambda_2h_1 &  18\lambda_1h_2 + 20\lambda_2h_2 &  -20\lambda_1h_3 - 14\lambda_2h_3 \\
     \end{pNiceArray},
    \end{align*}
and its determinant is given by
\begin{equation}
   \label{Eq:FirstCritPoly}
    \begin{aligned}
        q(h,\lambda) = &168\lambda_1^2h_1h_2 + 140\lambda_1^2h_1h_3 + 210\lambda_1^2h_2h_3 + 336\lambda_1\lambda_2h_1h_2+ 280\lambda_1\lambda_2h_1h_3\\
        + &98\lambda_1\lambda_2h_2h_3 + 168\lambda_2^2h_1h_2 + 140\lambda_2^2h_1h_3 - 112\lambda_2^2h_2h_3.
    \end{aligned}
    \end{equation}
\end{example}

{In \cite{PLOS_IdParaRegions} authors have shown that the signs attained by the critical polynomial provide valuable information about multistationarity.}  As a consequence of this result, Theorem~\ref{Thm::ReducedConnectivity} exploits the critical polynomial to establish the path connectivity of the parameter region of multistationarity.

\color{black}


\begin{theorem} \cite[Theorem 4]{ConnectivityPaper}
\label{Thm::ReducedConnectivity}
Consider a conservative consistent reaction network without relevant boundary steady states. Assume that there exist species $\mathrm{X}_{1}, \dots,\mathrm{X}_{k}$ such that each $\mathrm{X}_{j}$ participates in exactly $3$ reactions of the form
\[ \sum_{i=k+1}^m a_{i,j} X_{i} \xrightleftharpoons[\kappa_{3j}]{\kappa_{3j-1}} X_j  \xrightarrow{\kappa_{3j-2}} \sum_{i=k+1}^m b_{i,j} X_{i} , \quad j=1,\dots,k.\]
Let $q$ be the critical polynomial of the reduced network obtained by removing the reactions corresponding to $\kappa_{3j}$ for $j = 1,\dots ,k$. If
\begin{itemize}
    \item[(P1)] $q^{-1}(\mathbb{R}_{<0})$ is path connected, and 
    \item[(P2)] the Euclidean closure of $q^{-1}(\mathbb{R}_{<0})$ equals $q^{-1}(\mathbb{R}_{\leq0})$,
\end{itemize}
then the parameter regions of multistationarity of both the reduced and the original network are path connected.
\end{theorem}

\subsection{Computing the critical polynomial via a Gale dual matrix}\label{Sec::CritandGale}

{Finding the critical polynomial~\eqref{Eq::Critpolyexpression} requires computing the determinant of a symbolic matrix. This task becomes computationally infeasible if the matrix has many columns or if there are many symbolic variables. For instance, computing the critical polynomial for the weakly irreversible network~\eqref{Eq::Gunawardena} was not possible even for $n = 3$.  To simplify the computation of the critical polynomial, we derive a formula using \emph{Gale dual matrices} (Theorem \ref{Prop:CritPolyWnDn}). This approach is advantageous because the resulting matrices have a low number of columns. For the two families of phosphorylation networks considered, \eqref{Eq::nsitenetwork}  and \eqref{Eq::Gunawardena}, the corresponding Gale dual matrix has three columns for every $n$.}

For a set $I  \subseteq [k]$, we denote its complement by $I^c := [k] \setminus I$ and its cardinality by $\lvert I \rvert$. If $I = \{i_1 ,\dots , i_{p} \}  \subseteq [k]$ and $I^c = \{j_1, \dots , j_{k-p} \}$ with $i_1 < \dots < i_{p}$ and $j_1< \dots < j_{k-p}$, we define $\sgn( \tau_I ) \in \{ \pm 1\}$ to be the sign of the permutation $\tau_I$  that sends $(1, \dots , m)$ to $( i_1, \dots , i_{p}, j_1, \dots , j_{k-p})$. Additionally, consider any general matrix $Y\in K^{\ell_1\times \ell_2}$ and some set of indices $L_1 \subseteq [\ell_1]$ and $L_2 \subseteq [\ell_2]$. By $[Y]_{L_1,L_2} $ we denote the sub-matrix of $Y$ given by the rows and the columns indexed by the elements of $L_1$ and $L_2$ respectively.

{Next we rewrite $q(h,\lambda)$ using the Laplace expansion on complementary minors (see e.g. \cite[Theorem 2.4.1]{prasolov1994problems}) for the matrix $M(h,\l)$ in \eqref{Eq::Mntilde}. Since $M(h,\l)$ is expressed in terms of $W$ and $N'\diag(E\lambda)A^{\top}$, we can write the expansion in terms of minors of these matrices as
}

\begin{align}\label{Eq::Critpolyexpression} q(h,\lambda) = (-1)^{s}  \sum_{ \substack{I \subseteq [m] \\ |I| = s}} (-1)^{\sum_{j = d+1}^m j + \sum_{i \in I}i} \det \big( [W]_{[d],I^c} \big) \det\big( [N'\diag(E\lambda)A^{\top}]_{[s],I}\big) \prod_{i\in I}h_i.
	\end{align}
 
\begin{remark}  Viewed as a polynomial in $h$ , the coefficients of $q(h,\lambda)$ are given by maximal minors of $W \in \mathbb{R}^{d\times m}$ and of $N'\diag(E\lambda) A^{\top} \in\mathbb{R}(\lambda)^{s \times m} $. Therefore, if $\rk( N'\diag(E\lambda) A^{\top} ) < s$, then $q(h,\lambda)$ is the zero polynomial. 
\end{remark}

In the following, we treat $\lambda_1, \dots ,\lambda_\ell$ as symbolic variables and view $N'\diag(E\lambda) A^{\top}$ as a matrix with entries in the field of rational functions $\mathbb{R}(\lambda)$. Furthermore, we assume that $N'\diag(E\lambda) A^{\top}$ has full rank $s$. {A matrix $D(\lambda) \in \mathbb{R}(\lambda)^{m \times d}$ is called \emph{Gale dual} of $N'\diag(E\lambda) A^{\top}$ if  $\im(D(\lambda))=\ker(N'\diag(E\lambda) A^{\top})$ and $\ker(D(\lambda))=\{0\}$. To compute a Gale dual matrix, one simply needs to find a basis of $\ker(N'\diag(E\lambda) A^{\top})$ and write these vectors as the columns of $D(\lambda)$. }

\bigskip

\begin{example}
    In Example~\ref{Ex:Mhl}, we computed the matrix $N'\diag(E\lambda) A^{\top}$. It is given by 
    \begin{align*}
     N'\diag(E\lambda) A^{\top} =  \begin{pNiceArray}{ccc}
          4\lambda_1 + 4\lambda_2 &  -24\lambda_1 - 22\lambda_2 &  15\lambda_1 + 21\lambda_2 \\
           4\lambda_1 + 4\lambda_2 &  18\lambda_1 + 20\lambda_2 &  -20\lambda_1 - 14\lambda_2 \\
     \end{pNiceArray}.
    \end{align*}
    To find the kernel of this matrix, first we note that $v \in \mathbb{R}(\lambda)^3$ is contained in $\ker(N'\diag(E\lambda) A^{\top})$ if and only if $\diag(E\lambda) A^{\top}v \in \ker(N') = \ker(N)$. In Example~\ref{Ex:basis}, we computed a basis of $\ker(N)$. Thus, it follows that
    \begin{align}
    \label{Eq:CompGaleDual}
           \diag(E\lambda) A^\top v = \begin{pmatrix}
        4 \lambda_1 v_1 \\ 6(\lambda_1 + \lambda_2)v_2 \\ 5(\lambda_1+\lambda_2)v_3 \\ 4\lambda_2 v_1 + 2 \lambda_2 v_2 + 6 \lambda_2 v_3
    \end{pmatrix} =  \begin{pmatrix}
        \mu_1 \\ \mu_1 + \mu_2 \\ \mu_1 + \mu_2 \\ \mu_2
    \end{pmatrix} = E \mu  
    \end{align}
    for some $\mu \in \mathbb{R}(\lambda)^2$. By comparing the second and third entry of \eqref{Eq:CompGaleDual}, we have $v_2 = \tfrac{5}{6}$ if $v_3 = 1$. Now, adding the first and the last entry of \eqref{Eq:CompGaleDual} and comparing it with its third entry implies that $v_1 = \tfrac{15\lambda_1 - 8\lambda_2}{12(\lambda_1 + \lambda_2)}$. Since $\ker(N'\diag(E\lambda) A^{\top})$ has dimension one, it follows that
    \[ D(\lambda)^\top = \begin{pmatrix}
        \tfrac{15\lambda_1 - 8\lambda_2}{12(\lambda_1 + \lambda_2)} &
        \tfrac{5}{6} &
        1
    \end{pmatrix}\]
    In the proof of Theorem~\ref{Prop::D_nsite} and~\ref{Prop::D_Gunawardena}, we will use the same idea to compute the Gale dual matrices corresponding to the networks \eqref{Eq::nsitenetwork} and \eqref{Eq::Gunawardena}.
\end{example}

Using the maximal minors of a Gale dual matrix, one can be compute the coefficients of the critical polynomial as in \eqref{Eq::Critpolyexpression}
\color{black}

\begin{lemma}\label{lem:GDmatrix}
	\label{Lemma::GaleDualMinors}
	Let $D(\lambda) \in \mathbb{R}(\lambda)^{ m\times d}$ be a Gale dual matrix of $N'\diag(E\lambda) A^{\top} \in \mathbb{R}(\lambda)^{s \times m}$. There exists $\delta(\lambda) \in \mathbb{R}(\lambda) \setminus \{0\}$ such that for all $I  \subseteq [m]$  with $\lvert I \rvert = s$  it holds:
	\begin{align}
	\label{Eq::Delta}\delta(\lambda) \det \big( [D^\top(\lambda)]_{[d],I^c} \big) = \sgn(\tau_I) \det \big( [N'\diag(E\lambda) A^{\top} ]_{[s],I}\big).
	\end{align}
 {In particular, $\delta(\lambda) \in \mathbb{R}(\lambda)\setminus \{0\}$ is independent of $I \subseteq [m]$.}
\end{lemma}

\begin{proof}
The results holds for matrices over any field, so in particular over the rational function field $\mathbb{R}(\lambda)$ \cite[Theorem 12.16]{JoswigTheobald_book}, \cite[Lemma 2.10]{InjectivityPaper}.
\end{proof}

\color{black}
The above lemma now establishes a way to simplify the computation of the critical polynomial by computing the determinant of minors of size $d$ using the Gale dual matrices.

\begin{theorem}
	\label{Prop:CritPolyWnDn}
	Let $D(\lambda) \in \mathbb{R}(\lambda)^{m\times d}$ be a Gale dual matrix of $N' \diag(E \lambda )A^{\top} $. The critical polynomial \eqref{Eq:CritPoly} can be written as
	\[ q(h,\lambda) =   (-1)^{s(d+1)}\sum_{ \substack{I \subseteq [m] \\ \lvert I \rvert = s}} \delta(\lambda) \det \big( [W]_{[d],I^c} \big) \det \big( [D^\top(\lambda)]_{[d],I^c} \big) \prod_{i\in I}h_i ,\]
	where $\delta(\lambda) \in \mathbb{R}(\lambda) \setminus \{0\}$ satisfies \eqref{Eq::Delta}.
\end{theorem}

\begin{proof}
	Let $I = \{i_1 , \dots , i_{s} \} \subseteq [m]$, $I^c = \{j_1, \dots , j_{d} \} \subseteq [m]$ with $i_1 < \dots < i_{s}$ and $j_1< \dots < j_{d}.$ In the first part of the proof, we compute the sign of the permutation $\tau_I$. The number of inversions in $\tau$ is given by
 \begin{align}
 \label{Eq::SignTauI}
 v := \sum_{k=1}^d\big(m-j_k - (d-k) \big) =  ds - \sum_{k=1}^d j_k  + \sum_{k=1}^dk,
 \end{align}
 and therefore $\sgn(\tau_I) = (-1)^v$.
 
We substitute the expression of $\det \big( [N'\diag(E\lambda) A^{\top} ]_{[s],I}\big)$ from \eqref{Eq::Delta} in \eqref{Eq::Critpolyexpression}. The total power of $(-1)$ can, therefore, be computed as:
	\begin{align*}
	s + \sum_{k =d+1}^m k + \sum_{k = 1}^s i_k + v = s + 2 \sum_{k=1}^m k - 2 \sum_{k=1}^d j_k +ds \equiv s(d+1) \mod 2 .
	\end{align*}
 This concludes the proof.
\end{proof}

\begin{example}
    We compute $\delta(\lambda)$ from Lemma~\ref{Lemma::GaleDualMinors} for the running example. Since it is independent from $I \subseteq [3]$, we choose $I = \{1,2\}$. Computing the left and right-hand side of~\eqref{Eq::Delta} we have
    \begin{align*}
	&\delta(\lambda)\det \big( [D^\top(\lambda)]_{[1],\{3\}} \big) =  1 \quad \text{and}\\
 &\sgn(\tau_{\{1,2\}}) \det \big( [N'\diag(E\lambda) A^{\top} ]_{[2],\{1,2\}}\big) = 168\lambda_1^2 + 336\lambda_1\lambda_2 + 168\lambda_2^2,
	\end{align*}
 which gives $\delta(\lambda) = 168\lambda_1^2 + 336\lambda_1\lambda_2 + 168\lambda_2^2$. By Theorem~\ref{Prop:CritPolyWnDn}, the critical polynomial can be expressed as
 \begin{align*}
        q(h,\lambda) = \phantom{+} &\big(168\lambda_1 + 336\lambda_1\lambda_2 + 168\lambda_2^2\big)h_1h_2 + \big(140\lambda_1^2 + 280\lambda_1\lambda_2 + 140\lambda_2^2\big)h_1h_3 \\
        + &\big(210\lambda_1^2 + 98\lambda_1\lambda_2  - 112\lambda_2^2\big)h_2h_3.
    \end{align*}
Note that this is the same expression that we obtained in~\eqref{Eq:FirstCritPoly}.
\end{example}

\color{black}
\subsection{Signed supports of polynomials}\label{Section::SignedSupp}
To establish connectivity in the parameter region of multistationarity via Theorem~\ref{Thm::ReducedConnectivity}, it is enough to show that the properties (P1) and (P2) hold for the critical polynomial. 
{In the following, we recall methods from \cite{DescartesHypPlane,FaceReduction}, which we will apply to verify properties (P1) and (P2). These methods rely on the structure of the Newton polytope of the critical polynomial.}


Consider a multivariate polynomial function 
\[ f\colon \mathbb{R}^m_{>0} \to \mathbb{R}, \quad x \mapsto f(x) =\sum_{\mu \in \sigma(f) } c_\mu x^\mu,\]
where $ c_\mu \in \mathbb{R} \setminus \{0\}$ and $\sigma(f) \subseteq \mathbb{Z}^{m}$ is a finite set, called the \textit{support of $f$}. An exponent vector $\mu \in \sigma(f)$ is called \emph{positive} (resp. \emph{negative}) if the corresponding coefficient $c_\mu$ is positive (resp. negative). We write:
\[ \sigma_+(f) := \{ \mu \in \sigma(f) \mid c_\mu > 0 \} \quad \text{and} \quad \sigma_-(f) := \{ \mu \in \sigma(f) \mid c_\mu < 0 \}.\] The Newton polytope of $f$, denoted by $\N(f)$ is given by the convex hull of $\sigma(f)$. For $S \subseteq \mathbb{R}^m$, we denote the \emph{restriction} of $f$ to $S$ by
\[ f_{|S}(x) = \sum_{\mu \in \sigma(f) \cap S} c_\mu x^\mu.\]

\bigskip

Following \cite{FaceReduction}, we say that $f$ has \emph{one negative connected component} if 
\[ f^{-1}(\mathbb{R}_{<0}) = \big\{ x \in \mathbb{R}^n_{>0} \mid f(x) < 0\big\} \]
is a connected set. If the Euclidean closure of $f^{-1}(\mathbb{R}_{<0})$ equals $f^{-1}(\mathbb{R}_{\leq0})$, then $f$ \emph{satisfies the closure property}. These are the terminologies for conditions (P1) and (P2) in Theorem~\ref{Thm::ReducedConnectivity}.

For fixed $x \in \mathbb{R}^m_{>0}$ {such that $f(x)<0$}, each $v \in \mathbb{R}^m \setminus \{ 0 \}$ induces a path
\begin{align}
\label{Eq:PathVt}
 [1,\infty) \to \mathbb{R}^m_{>0},\qquad  t \mapsto t^v \ast x = (t^{v_1} x_1, \dots , t^{v_m} x_m).
\end{align}
This path is contained in $f^{-1}(\mathbb{R}_{<0})$ if the univariate polynomial $f(t^v\ast x)$ does not have a root for $t \geq 1$. Using Descartes' rule of signs {\cite[Corollary 10.1.10]{rahman2002analytic}}, it is easy to see that this is the case if $f(t^v \ast x)$ has one sign change in its coefficient sequence and its leading coefficients is negative. Motivated by this observation, we introduce \emph{separating hyperplanes} of the support of $f$.

\color{black}
For $v \in \mathbb{R}^m \setminus \{ 0 \}$ and $a \in \mathbb{R}$, we define the \emph{hyperplane} 
$\mathcal{H}_{v,a} := \{ \mu \in \mathbb{R}^{m} \mid v \cdot \mu = a \}$, and the following two \emph{half-spaces}:
\[ \mathcal{H}^+_{v,a} = \{ \mu \in \mathbb{R}^{m} \mid v \cdot \mu \geq a \}, \qquad \mathcal{H}^-_{v,a} = \{ \mu \in \mathbb{R}^{m} \mid v \cdot \mu \leq a \}.\]
We will denote by $\mathcal{H}^{+,\circ}_{v,a}$ and $\mathcal{H}^{-,\circ}_{v,a}$ the two \emph{open half-spaces}:
\[ \mathcal{H}^{+,\circ}_{v,a} = \{ \mu \in \mathbb{R}^{m} \mid v \cdot \mu > a \}, \qquad \mathcal{H}^{-,\circ}_{v,a} = \{ \mu \in \mathbb{R}^{m} \mid v \cdot \mu < a \}.\]
A hyperplane $\mathcal{H}_{v,a}$ is called a \emph{separating hyperplane} of $\sigma(f)$ if
$ \sigma_-(f) \subseteq \mathcal{H}^+_{v,a}$ and $ \sigma_+(f) \subseteq \mathcal{H}^-_{v,a}.$
Additionally, we call $\mathcal{H}_{v,a}$ a \emph{strict separating hyperplane} if there exists some $\mu \in \sigma_-(f)$ such that $v \cdot \mu > a$. {In that case, the univariate polynomial $f(t^v \ast x)$ has one sign change in its coefficient sequence and its leading coefficients is negative. Thus, one can build continuous paths, such as~\eqref{Eq:PathVt}, in $f^{-1}(\mathbb{R}_{<0})$ explicitly.} The following proposition now recalls a result in \cite{DescartesHypPlane}, that establishes the properties (P1) and (P2) of $f$ based on the existence of a strict separating hyperplane.

\begin{proposition}
	\label{Thm::SepHypThm}
	\cite[Theorem 3.6]{DescartesHypPlane} 
  For a polynomial function $f\colon \mathbb{R}^m_{>0} \to \mathbb{R}, \,  f(x) =\sum_{\mu \in \sigma(f) } c_\mu x^\mu$, if $\sigma(f)$ has a strict separating hyperplane, then $f$ has one negative connected component and satisfies the closure property.
\end{proposition}

\begin{example}
    Consider the critical polynomial $q(h,\lambda)$ of the running example from~\eqref{Eq:FirstCritPoly}. It has only one negative coefficient corresponding to the monomial $-112 \lambda_2^2 h_2 h_3$. Thus, $q(h,\lambda)$ has one negative exponent vector $\beta = (0,2,0,1,1)$. For $v= (0,1,0,1,1)$ and all $\mu \in \sigma(q) \setminus \{ \beta \}$, we have $v \cdot \beta = 4 > v \cdot \mu$. Therefore, $\mathcal{H}_{v,4}$ is a strict separating hyperplane of $\sigma(q)$, and consequently $q$ satisfies (P1) and (P2) by Theorem~\ref{Thm::SepHypThm}. Using Theorem~\ref{Thm::ReducedConnectivity}, we conclude that the parameter region of multistationarity for the running example is connected.
\end{example}

The supports of the critical polynomials corresponding to the networks \eqref{Eq::nsitenetwork}  and \eqref{Eq::Gunawardena} do not have a strict separating hyperplane. We recall two additional theorems that we will use in the proofs in Sections~\ref{Section::SIP} and \ref{Section::Weakly}. The proof of these results rely on constructing explicit path in $f^{-1}(\mathbb{R}_{<0})$ using the paths of the form~\eqref{Eq:PathVt}.

\color{black}

For a polytope $P \subseteq \mathbb{R}^m$, the \emph{face} with normal vector $v \in \mathbb{R}^m$ is given by
\[ P_v := \big\{ \mu \in P \mid v \cdot \mu = \max_{\nu \in P} v \cdot \nu \big\}. \]
We call two faces $P_v, P_w \subseteq P$ \emph{parallel} if $v = -w$. An \emph{edge} (resp. \emph{vertex}) of $P$ is a face of dimension $1$ (resp. $0$). We write $\Vertex(P)$ for the set of vertices of $P$. For a Newton polytope, we denote the face with normal vector $v$ by $\N_v(f)$. If $F \subseteq \N(f)$ is an edge and $F \cap \sigma(f) \subseteq \sigma_-(f)$, we call $F$ a \emph{negative edge of $\N(f)$}.

\begin{theorem}
	\label{Thm::NegativeFace}
	\cite[Theorem 3.1]{FaceReduction}
 Let $f\colon \mathbb{R}^m_{>0} \to \mathbb{R}, \,  f(x) =\sum_{\mu \in \sigma(f) } c_\mu x^\mu$ be a polynomial function. If there exists a proper face $\N_v(f) \subsetneq \N(f)$ such that $\sigma_-(f) \subseteq \N_v(f)$, then $f$ satisfies the closure property. If additionally $f_{|\N_v(f)}$ has one negative connected component, then $f$ also has one negative connected component.
\end{theorem}

The following result from \cite{FaceReduction} provides a condition for splitting the polynomial into two parts and establishes that if both smaller polynomials have one negative connected component, then so does the original polynomial. 

\begin{theorem}
	\label{Thm::ParallelFaces}
	\cite[Theorem 3.6]{FaceReduction}
Let $f\colon \mathbb{R}^m_{>0} \to \mathbb{R}, \,  f(x) =\sum_{\mu \in \sigma(f) } c_\mu x^\mu$ be a polynomial function. Assume that there exist parallel faces $\N_v(f), \, \N_{-v}(f) \subseteq \N(f)$ such that $\sigma(f) \subseteq \N_v(f) \cup \N_{-v}(f)$ and both $f_{|\N_v(f)}$ and $f_{|\N_{-v}(f)}$ have one negative connected component.
	If there exist $\mu_0 \in \N_v(f) \cap \sigma_-(f)$ and $\mu_1 \in \N_{-v}(f)  \cap \sigma_-(f)$  such that $\Conv(\mu_0,\mu_1)$ is an edge of $\N(f)$, then $f$ has one negative connected component.
\end{theorem}

{In the next example, we demonstrate how Proposition~\ref{Thm::SepHypThm}, Theorem~\ref{Thm::NegativeFace} and \ref{Thm::ParallelFaces} can be combined to show that a polynomial has one negative connected component} 
\begin{example}
\label{Ex::Faces}
Let $c_1, \dots ,c_7 \in \mathbb{R}_{>0}$ and consider the polynomial 
\[ f = c_1 x_1 + c_2 x_1 x_2 - c_3 x_2 - c_4 - c_5 x_1 x_3 - c_6 x_1 x_2 x_3 + c_7 x_2 x_3,\]
which has $3$ positive and $4$ negative exponent vectors:
\[ \sigma_+(f) = \{ (1,0,0), (1,1,0), (0,1,1) \}, \quad \sigma_-(f) = \{ (0,0,0),(0,1,0),(1,0,1),(1,1,1) \}.\]
These exponent vectors are shown in red and blue colour respectively in Figure \ref{FIG::faces}. For $e_3 = (0,0,1)$, the faces 
\begin{align*}
F &= \N_{-e_3}(f)  = \Conv( (0,0,0),(1,0,0),(0,1,0),(1,1,0) ),\\
G &= \N_{e_3}(f)  = \Conv( (1,0,1),(0,1,1),(1,1,1) )
\end{align*}are parallel and their union contains $\sigma(f)$. The restricted polynomials are given by
\[ f_{|F} = c_1 x_1 + c_2 x_1 x_2 - c_3 x_2 - c_4 , \qquad f_{|G} =  - c_5 x_1 x_3 - c_6 x_1 x_2 x_3 + c_7 x_2 x_3 .\]
Since the hyperplanes $\mathcal{H}_{-e_1,-0.5}$ and $\mathcal{H}_{e_1,0.5}$  with $e_1 = (1,0,0)$ are strict separating hyperplanes of $\sigma(f_{|F})$ and $\sigma(f_{|G})$ respectively, Proposition \ref{Thm::SepHypThm} implies that both $f_{|F}$ and $f_{|G}$ have one negative connected component.

For the negative exponent vectors $\mu_0 = (0,0,0), \, \mu_1 = (1,0,1)$, $\Conv(\mu_0,\mu_1)$ is an edge of $\N(f)$. Using Theorem \ref{Thm::ParallelFaces}, we conclude that $f$ has one negative connected component. For an illustration, we refer to Figure \ref{FIG::faces}.
\end{example}

\begin{figure}[t]
\centering
\includegraphics[scale=0.40]{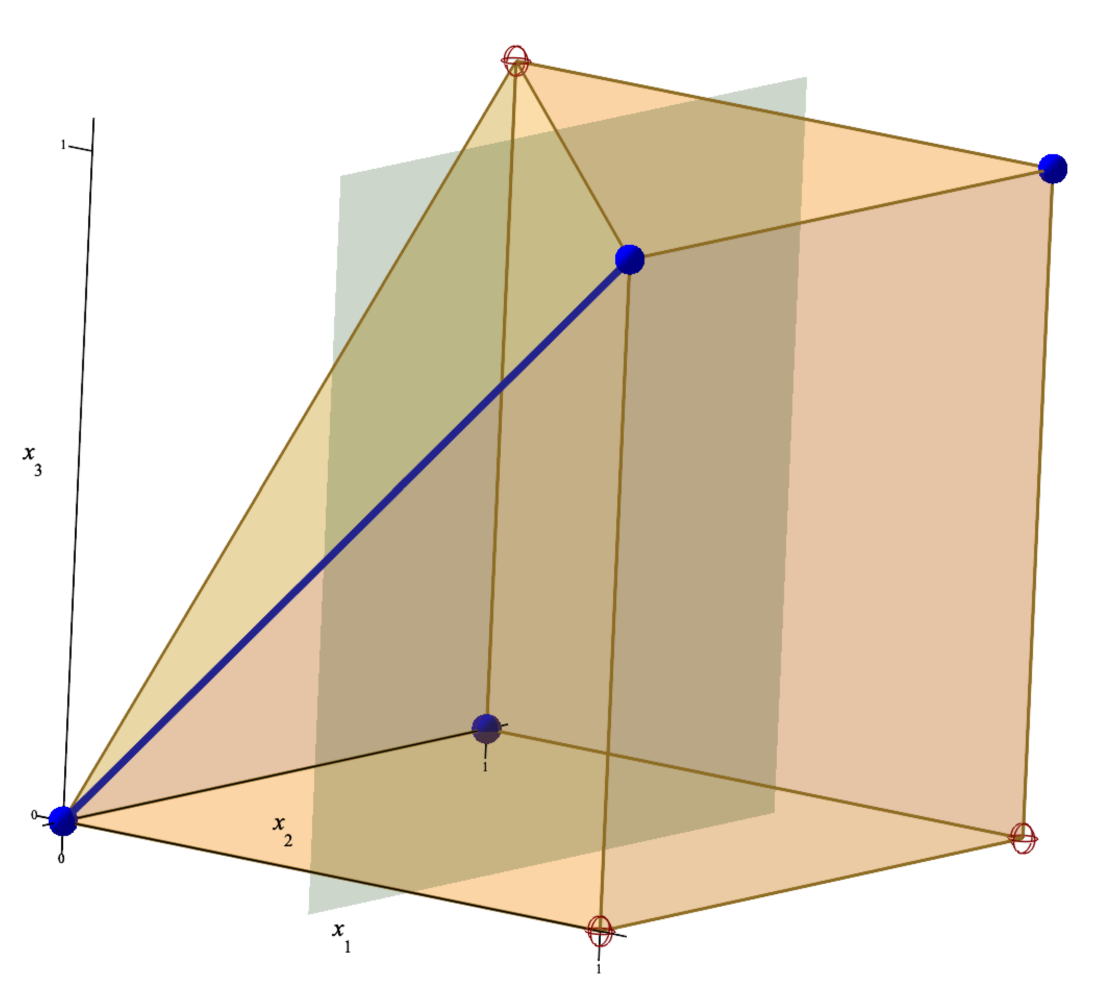}
\caption{{\small An illustration of Example \ref{Ex::Faces}. The positive and negative exponent vectors $\sigma_+(f), \, \sigma_-(f)$ are marked with red circles and blue dots respectively. The grey hyperplane with normal vector $e_1 = (1,0,0)$ is a strict separating hyperplane of $\sigma(f_{|F})$ and $\sigma(f_{|G})$, for $F = \N_{-e_3}(f)$,  $G = \N_{e_3}(f)$.  The blue thick edge $\Conv((0,0,0),(1,0,1))$ connects negative exponent vectors of $f_{|F}$ and $f_{|G}$}}\label{FIG::faces}
\end{figure}

\subsection{Finding edges of the Newton polytope}\label{Sec::EdgesNP}
{In this section, we provide two technical results} that will be used in Section \ref{Section::Proofnsite} and 
\ref{Section::ProofGunawardena}. In Section \ref{Section::Proofnsite}, we deal with polynomial functions with the special property that $\sigma(f) \subseteq \{ 0,1 \}^m$ and each $\mu \in \sigma(f)$ has exactly $d$ zero entries, for some fixed $d\in \mathbb{N}$. For $I \subseteq [m]$ with $\lvert I \rvert  = d$, we write $z_I\subset \RR^m$ for the vector with $(z_I)_i = 0$ for $i \in I$ and $(z_I)_i = 1$ for $i \notin I$. Moreover, we denote by $e_1, \dots ,e_m \in \mathbb{R}^m$ the standard basis vectors in $\mathbb{R}^m$.

\begin{proposition}
	\label{Prop::NegEdge}
	For $d, m \in \mathbb{N}$ with $d < m$, let $P \subseteq \mathbb{R}^m$ be a polytope such that \[\Vertex(P) \subseteq \big\{ z_I  \in \{ 0,1 \}^m \mid I \subseteq [m], \, |I| = d \big\},\]
    and let $J_1,J_2 \subseteq [m]$ such that $\lvert J_1 \rvert = \lvert J_2 \rvert = d$. If $\lvert J_1\cap J_2 \rvert = d-1$ and $z_{J_1},z_{J_2} \in P$, then $\Conv(z_{J_1},z_{J_2})$ is an edge of $P$.
\end{proposition}

\begin{proof}
	Since $\lvert J_1\cap J_2 \rvert = d-1$, we have that $J_1 = (J_1\cap J_2) \cup \{j_1\}, J_2 = (J_1\cap J_2) \cup \{j_2\}$ for $j_1 \neq j_2 $. Let $v := e_{j_1}+e_{j_2}+2 \sum_{i \in (J_1\cup J_2)^c} e_i $. For every $z_I \in \mathbb{R}^m$, it holds:
\begin{equation}
\label{Eq:Prop1_zI}
 \begin{aligned}
	\begin{split}
	v \cdot z_I  = 2|(J_1\cup J_2)^c| + 1 = 2(m-d-1)+1, \quad \text{if } I = J_1 \text{ or } I = J_2 \\
	v \cdot z_I  < 2|(J_1\cup J_2)^c |+ 1 = 2(m-d-1)+1, \quad \text{if } I \neq J_1 \text{ or } I \neq J_2.
	\end{split}
	\end{aligned}
\end{equation}
From \eqref{Eq:Prop1_zI}, it follows that $v \cdot \mu \leq 2(m-d-1)+1$  for all $\mu \in P$ with equality if and only if $\mu \in  \Conv(z_{J_1},z_{J_2})$. Thus, $\Conv(z_{J_1},z_{J_2})$ is an edge of $P$.
\end{proof}

In Section \ref{Section::ProofGunawardena}, we will work with polynomial functions $f$ such that $\sigma(f) \subseteq \{0,1\}^m \times \mathbb{R}^{\ell}$. To find edges of the Newton polytope for such polynomials the following proposition will be particularly helpful.

\begin{proposition}
\label{Prop::EdgeLifting}
Let $\pr_1\colon \mathbb{R}^{m} \times \mathbb{R}^{\ell} \to \mathbb{R}^{m}$ and $\pr_2\colon \mathbb{R}^{m} \times \mathbb{R}^{\ell} \to \mathbb{R}^{\ell}$ be projections onto the first $m$ and onto the last $\ell$ coordinates respectively.
Let $P \subseteq \mathbb{R}^{m} \times \mathbb{R}^\ell$ be a polytope, $x_1,x_2 \in \Vertex( \pr_1(P) )$ and $y \in \Vertex( \pr_2(P) )$. If $\Conv(x_1,x_2)$ is an edge of $\pr_1(P)$ such that 
\[ \pr_1( \Vertex(P) ) \cap \Conv(x_1,x_2) = \{ x_1, x_2 \},\] 
then $\Conv\big( (x_1,y),(x_2,y) \big)$ is an edge of $P$.
\end{proposition}

\begin{proof}
Let $v \in \mathbb{R}^{m}$ be a normal vector of the face $\Conv(x_1,x_2) \subseteq \pr_1( \Vertex(P) )$. Since $\pr_1( \Vertex(P) ) \cap \Conv(x_1,x_2) = \{ x_1, x_2 \}$, for $ z \in \Vertex(P)$  we have
\[ v \cdot \pr_1(z) \leq v \cdot x_1 = v \cdot x_2, \]
with equality if and only if $\pr_1(z) \in \{ x_1, x_2\}$. 

Let $w \in \RR^\ell$ be a normal vector of $\Vertex(\pr_2(P))=\{y\}$. For $z \in \Vertex(P)$, it holds
\[w \cdot \pr_2(z) \leq w \cdot y \]
 with equality if and only if $\pr_2(z) = y$.

From the above inequalities, it follows that for $z \in \Vertex(P)$
\[ (v,w) \cdot z \leq (v,w) \cdot (x_1, y) = (v,w) \cdot (x_2, y) \]
with equality if and only if $z = (x_1,y)$ or $z = (x_2,y)$. So, $\Conv( (x_1,y),(x_2,y))$ is an edge of~$P$.
\end{proof}

\subsection{Overview of the approach}
\label{Section:Overview}

In this section, we briefly summarize the approach we will use in Sections \ref{Section::SIP} and \ref{Section::Weakly} to establish connectivity of the parameter region of multistationarity in the networks \eqref{Eq::nsitenetwork} and \eqref{Eq::Gunawardena}, for every $n \geq 2$. We aim for this method to be applicable to other families of reaction networks. Therefore, we outline the key steps and how to address them.

Our arguments rely on Theorem \ref{Thm::ReducedConnectivity}, which applies to conservative consistent reaction networks without relevant boundary steady states. Both network families \eqref{Eq::nsitenetwork} and \eqref{Eq::Gunawardena} are post-translational modification networks, which are conservative and consistent. By \cite[Corollary 2]{CatalystInterm} both networks do not have any relevant boundary steady states. So we can apply Theorem~\ref{Thm::ReducedConnectivity} {for the species $Y_{1},\dots,Y_{n},U_1, \dots U_n$ and $Y_1,Y_3, \dots Y_{2n-1},U_1,U_3, \dots U_{2n-1}$ respectively} and only consider the networks obtained from \eqref{Eq::nsitenetwork} (resp. \eqref{Eq::Gunawardena}) by removing the reversible reactions corresponding to $\kappa_{6i+2}, \kappa_{6i+5}$  (resp. $\kappa_{10i+2}, \kappa_{10i+7}$), $i = 0, \dots , n-1$. After this modification, the reduced strongly irreversible phosphorylation network is given by
\begin{align}\tag{$\mathcal{F}_{1,n}$}
\begin{split}
\label{Eq::nsitenetwork_reduced}
S_{i} +K \xrightarrow{\kappa_{4i+1}} &Y_{i+1} \xrightarrow{\kappa_{4i+2}} S_{i+1} + K \\
 S_{i+1} +F \xrightarrow{\kappa_{4i+3}} &U_{i+1} \xrightarrow{\kappa_{4i+4}} S_{i} + F, 
\end{split}
\qquad \quad i =0,\dots,n-1,
\end{align}
and the reduced weakly irreversible phosphorylation network has the form
\begin{align}
\label{Eq::Gunawardena_reduced}\tag{$\mathcal{F}_{2,n}$}
\begin{split}
S_{i} +K  \xrightarrow{\kappa_{8i+1}}  &Y_{2i+1}  \xrightarrow{\kappa_{8i+2}}   Y_{2i+2}  \xrightleftharpoons[\kappa_{8i+4}]{\kappa_{8i+3}}  S_{i+1} + K\\ \nonumber
S_{i+1} +F \xrightarrow{\kappa_{8i+5}}  &U_{2i+1} \xrightarrow{\kappa_{8i+6}} U_{2i+2} \xrightleftharpoons[\kappa_{8i+8}]{\kappa_{8i+7}}S_{i} + F, 
\end{split}\qquad i =0,\dots,n-1.
\end{align}

By Theorem~\ref{Thm::ReducedConnectivity}, it is enough to show that the critical polynomial $q_n$ of the reduced network satisfies properties (P1) and (P2), that is, it has one negative connected component and satisfies the closure property. To that end, first in Sections \ref{Section::SIPcriticalPol} and \ref{Section::Gunawardena}, we provide recursive formulas for the stoichiometric matrix $N_n \in \mathbb{R}^{m \times r}$, reactant matrix $A_n \in \mathbb{R}^{m \times r}$, a conservation law matrix $W_n \in \mathbb{R}^{d \times m}$ and an extreme matrix $E_n \in \mathbb{R}^{r \times \ell}$ of the networks \eqref{Eq::nsitenetwork_reduced} and \eqref{Eq::Gunawardena_reduced}. Moreover, we compute a Gale dual matrix $D_n(\lambda) \in \mathbb{R}(\lambda)^{m \times d}$ of $N_n' \diag(E_n \lambda )A_n^{\top} \in \mathbb{R}(\lambda)^{s \times m}$. We use these to write $q_n$ in a recursive form.

For both families \eqref{Eq::nsitenetwork_reduced} and \eqref{Eq::Gunawardena_reduced}, we have $d = 3$ for every $n \in \mathbb{N}$. Thus, using Theorem \ref{Prop:CritPolyWnDn}, we can compute the coefficients of the critical polynomial $q_n$ by computing minors of size $3$ of the matrices $W_n$ and $D_n(\lambda)$. In Sections \ref{Section::Proofnsite} and \ref{Section::ProofGunawardena}, we conduct these computations, focusing on the signs of the coefficients of $q_n$. 
For $n=1$, the critical polynomial $q_1$ has only positive coefficients for both \eqref{Eq::nsitenetwork_reduced} and \eqref{Eq::Gunawardena_reduced}. Thus, the parameter region of multistationarity is empty by \cite[Theorem 1]{PLOS_IdParaRegions}. For $n=2$, the polynomial $q_2$ satisfies (P1) and (P2) according to \cite{ConnectivityPaper} for \eqref{Eq::nsitenetwork_reduced}, and as shown in \cite{FaceReduction} for \eqref{Eq::Gunawardena_reduced}.

To prove that $q_n, n \geq 3$ satisfies the closure property, we show that $\sigma_-(q_n)$ is contained in a proper face $F \subsetneq \N(q_n)$ (cf. Theorem~\ref{Thm::NegativeFace}). It is now enough to show that $q_{n|F}$ satisfies the property~(P1). We split up $q_{n|F}$ into sub-polynomials based on parallel faces of the Newton polytope (cf. Theorem~\ref{Thm::ParallelFaces} and Example~\ref{Ex::Faces}). One of these sub-polynomials corresponds to $q_{n-1}$.  For the remaining sub-polynomials, we show that their signed support has a strict separating hyperplane, which implies that they have one negative connected component (cf. Proposition~\ref{Thm::SepHypThm}). To conclude that $q_n|_F$ and consequently, $q_n$ has one negative connected component we will use induction on $n$ and Theorem~\ref{Thm::ParallelFaces}. To identify negative edges between parallel faces, as required in Theorem~\ref{Thm::ParallelFaces}, we will apply Proposition \ref{Prop::NegEdge} and \ref{Prop::EdgeLifting}.

Combining these arguments shows that the critical polynomial $q_n$ of the reduced networks satisfy properties (P1) and (P2). Consequently, Theorem~\ref{Thm::ReducedConnectivity} implies that the parameter region of multistationarity is connected for both the reduced and the original network.

\section{\bf Strongly Irreversible Phosphorylation Networks}\label{Section::SIP}

 In this section, we study the family of reduced strongly irreversible phosphorylation networks \eqref{Eq::nsitenetwork_reduced}. First in Section~\ref{Section::SIPcriticalPol} we compute their critical polynomial and then in Section~\ref{Section::Proofnsite} we use it to show that the parameter region of multistationarity is connected for all $n\geq 2$.

\subsection{Computation of Critical Polynomial}
\label{Section::SIPcriticalPol}
To define the matrices $N_n$ and $A_n$ for \eqref{Eq::nsitenetwork_reduced}, we use the order $K,F,S_0,Y_1,U_1,S_1,  \dots,Y_n,U_n,S_n$ for the species and $\kappa_1, \dots, \kappa_{4n}$ for the reactions to label the rows and the columns, respectively, of $N_n$ and $A_n$. We now give the matrix for $N_1$ and a recursive expression of the stoichiometric matrices $N_n\in \RR^{(3n+3)\times 4n}$:
\begin{align} 
\label{Eq::N_recursive_nsite}
N_1 = \begin{pNiceArray}{c}
                   P_1\\
                   P_2 \end{pNiceArray} \in \mathbb{R}^{6 \times 4} \qquad \text{ and } \qquad 
                          N_n =\begin{pNiceArray}{cc}
                          N_{n-1}&P_1\\
                          {~~ 0_{3\times (4n-4)}}&P_2
                          \end{pNiceArray}\in \RR^{(3n+3)\times 4n},
                          \end{align}
where the $(3n\times 4)$-matrix $P_1$ and $(3\times 4)$-matrix $P_2$ are given by
\begin{align}
\label{Eq::P1P2_nsite}
    P_1 =\begin{pNiceArray}{cccc}
         -1&\phantom{-}1&\phantom{-}0&\phantom{-}0\\
         \phantom{-}0&\phantom{-}0&-1&\phantom{-}1\\
         \phantom{-}0&\phantom{-}0&\phantom{-}0&\phantom{-}0\\
         \phantom{-}\vdots&\phantom{-}\vdots&\phantom{-}\vdots&\phantom{-}\vdots\\
         \phantom{-}0&\phantom{-}0&\phantom{-}0&\phantom{-}0\\
         -1&\phantom{-}0&\phantom{-}0&\phantom{-}1\\
     \end{pNiceArray} \in \RR^{3n\times 4} \qquad \text{ and } \qquad P_2 =\begin{pNiceArray}{cccc}
         \phantom{-}1&-1&\phantom{-}0&\phantom{-}0\\
         \phantom{-}0&\phantom{-}0&\phantom{-}1&-1\\
         \phantom{-}0&\phantom{-}1&-1&\phantom{-}0\\
     \end{pNiceArray} \in \RR^{3\times 4}.
 \end{align}

 \begin{remark}\label{rem::Nnrank}
     We point out that matrix $P_1$ has exactly $3n-3$ rows {(rows $3,\ldots,3n-1$)} with only zero entries. Matrix $P_2$ has full rank and the rank of $N_1$ is 3. From the recursive relation it is easy to deduce that $\rank(N_n)=3n.$ In particular, $\dim(\im(N_n))=3n$ and the dimension of the left kernel of $N_n$ is $3$.
 \end{remark}

A recursive expression for the reactant matrices $A_n\in \RR^{(3n+3)\times 4n}$ is given as
\begin{align} 
\label{Eq::A_recursive_nsite}
A_1 = \begin{pNiceArray}{c}
                            Q_1\\
                            Q_2 \end{pNiceArray} \in \mathbb{R}^{6 \times 4} \qquad \text{ and } \qquad  A_n= \begin{pNiceArray}{cc}
                          A_{n-1}&Q_1\\
                          ~~~0_{3\times (4n-4)}&Q_2
                          \end{pNiceArray}\in \mathbb{R}^{(3n+3) \times 4n},
                          \end{align}
where the matrices $Q_1$ and $Q_2$ are given by
\begin{align}
\label{Eq::Q1Q2_nsite}
    Q_1 =\begin{pNiceArray}{cccc}
         1&0&0&0\\
         0&0&1&0\\
         0&0&0&0\\
         \vdots&\vdots&\vdots&\vdots\\
         0&0&0&0\\
         1&0&0&0\\
     \end{pNiceArray}\in \RR^{3n\times 4} \qquad \text{ and } \qquad 
     Q_2 =\begin{pNiceArray}{cccc}
         0&1&0&0\\
         0&0&0&1\\
         0&0&1&0\\
     \end{pNiceArray}\in \RR^{3\times 4}.
 \end{align}

 Same as $P_1$, matrix $Q_1$ has exactly $3n-3$ rows {(rows $3,\ldots,3n-1$)} with zero entries. With the stoichiometric matrix as in \eqref{Eq::N_recursive_nsite}, we now determine a conservation law matrix $W_n.$ 

\begin{lemma}
\label{Lemma_Wn_nsite}
	For the stoichiometric matrix $N_n$ in \eqref{Eq::N_recursive_nsite}, a conservation law matrix $W_n$ is given by:
\begin{align}\label{Eq::Wn_nsite}
	W_n:=\begin{pmatrix}
	\Id_3 & \smash[b]{\blockW{n}}
\end{pmatrix} \in \RR^{3\times(3n+3)},
\end{align}
\medskip

\noindent where $\Id_3$ is the identity matrix of size 3 and
	\[\overline{W}:=\begin{pNiceArray}{ccc}
	1&0&0\\
	0&1&0\\
	1&1&1
	\end{pNiceArray}\in \RR^{3\times 3}.\]
\end{lemma}

\begin{proof}
Since the dimension of the left kernel of $N_n$ is $3$ (Remark~\ref{rem::Nnrank}) and $W_n$ has rank $3$ for all $n$, it is enough to show that $W_n N_n = 0$. 
	Using $W_0:=\Id_3$ and $W_n=\begin{pNiceArray}{cc}
	W_{n-1}&\overline{W}
	\end{pNiceArray}$, for $n\geq 2$ we have
	\[W_n N_n=\begin{pNiceArray}{cc}
	W_{n-1}N_{n-1}&W_{n-1}P_1+\overline{W}P_2\end{pNiceArray}.\]
	Simple computations give
	\begin{align*}
	W_{n-1}P_1 =\begin{pNiceArray}{cccc}
	-1&\phantom{-}1&\phantom{-}0&\phantom{-}0\\
	\phantom{-}0&\phantom{-}0&-1&\phantom{-}1\\
	-1&\phantom{-}0&\phantom{-}0&\phantom{-}1
	\end{pNiceArray} \in \RR^{3\times 4} \qquad \text{ and } \qquad \overline{W}P_2 =\begin{pNiceArray}{cccc}
	\phantom{-}1&-1&\phantom{-}0&\phantom{-}0\\
	\phantom{-}0&\phantom{-}0&\phantom{-}1&-1\\
	\phantom{-}1&\phantom{-}0&\phantom{-}0&-1\\
	\end{pNiceArray} \in \RR^{3\times 4}.
	\end{align*}
 Thus, for $n=1$, we get $W_1 N_1=W_{n-1}P_1+\overline{W}P_2=0$. 
	Let us now assume that the result holds for $k=n-1$ for some $n\geq 1.$ By assumption $W_{n-1}N_{n-1}=0_{3\times (4k-4)}.$
	Consequently, $W_n N_n=0_{3\times 4k}$  and therefore, $W_n$ is a conservation law matrix for $N_n$.
\end{proof}
 
\begin{proposition}\label{Lemma::ExVectorsReduced_nsite}
An extreme matrix of \eqref{Eq::nsitenetwork_reduced} has the following recursive form
\begin{align}
\label{Eq::En_nsite} 
E_n = 
\begin{pNiceArray}{c | c}
E_{n-1} & 0_{(4n-4) \times 1}\\ \hdottedline
0_{4\times (n-1)}& E_1\end{pNiceArray} \in \mathbb{R}^{4n\times n},
\end{align}
where 
\[ E_1^{\top} = \begin{pNiceArray}{cccc}
1 & 1 & 1& 1
\end{pNiceArray}\in \RR^{1\times 4}. \]
Moreover, the $n$ column vectors $E_n^{(1)}, \dots , E_n^{(n)} \in \mathbb{R}^{4n}$ of $E_n$ form a basis of $\ker(N_n).$
\end{proposition}

\begin{proof}
Using \eqref{Eq::N_recursive_nsite}, it is easy to check that $E_n^{(1)}, \dots ,E_n^{(n)}$ are contained in $\ker(N_n) \cap \mathbb{R}^{4n}_{\geq 0}$. Since 
\[ \dim \ker(N_n) = 4n - \dim \im (N_n) = 4n - 3n = n\]
(cf. Remark~\ref{rem::Nnrank}) and $E_n^{(1)}, \dots ,E_n^{(n)}$ are linearly independent, they form a basis of $\ker(N_n)$. Moreover, $E_n^{(1)}, \dots ,E_n^{(n)}$ have pairwise disjoint support and hence, by Proposition \ref{Prop::ExtremeVector} $E_n$ gives an extreme matrix for the network \eqref{Eq::nsitenetwork_reduced}.
\end{proof}

The matrix $N_n'$ is obtained by removing the first three rows of $N_n.$ In the remainder of this subsection, we compute a Gale dual matrix $D_n(\lambda)$ of $N_n'\diag(E_n\l)A_n^{\top} \in \mathbb{R}(\lambda)^{3n \times (3n+3)}$ and $\delta_n(\lambda)$ as in Theorem \ref{Prop:CritPolyWnDn} for the network \eqref{Eq::nsitenetwork_reduced}, considering  $\l=(\l_1,\ldots,\l_n)$ as symbolic variables.

In Lemma~\ref{lem::detCn}, we compute the maximal minor $\det( [ N_n'\diag(E_n\lambda)A_n^{\top} ]_{[s],I})$ for $I = \{4,\dots,3n+3\}$ and $I^c = \{1,2,3\}$. 
\begin{lemma}
    \label{lem::detCn}
    Given $n \in \mathbb{N}$, let $s=3n$, $I=\{4,\ldots,3n+3\}$. Then,
    \[\det( [ N_n'\diag(E_n\lambda)A_n^{\top} ]_{[s],I})=(-1)^{n}\prod_{i=1}^n \lambda_i^3.\]
   In particular, $\ker(N_n'\diag(E_n\lambda)A_n^{\top})$ has dimension $3$.
\end{lemma}
\begin{proof}
    Note that the matrix  $[ N_n'\diag(E_n\lambda)A_n^{\top} ]_{[s],I}$ is obtained by deleting the first 3 columns of  $N_n'\diag(E_n\lambda)A_n^{\top}$. Let $A_n'$ denote the matrix obtained from $A_n$ after deleting the first three rows. Then,
	\[ [ N_n'\diag(E_n\lambda)A_n^{\top} ]_{[s],I} = N_n'\diag(E_n\lambda)A_n^{'\top}.\]
We define the matrices $C_n \in \mathbb{R}(\lambda)^{3n \times 3n}$ recursively as:
{\small	\begin{align*}
	C_n := 
	\begin{pNiceArray}{c c c| c}
	&C_{n-1}&&0_{3\times 3}\\ \hdottedline
	&0_{3\times (3n-6)}&V_n&Z_n\end{pNiceArray} \in \mathbb{R}(\lambda)^{3n\times3n} \quad \text{ and } \quad
 C_1:=Z_1
	\end{align*}
}
where 
{\small	\begin{align*}
	V_n := \begin{pNiceArray}{c c c }
	0 & 0 & \lambda_n \\
	0 & 0 & 0 \\ 
	0 & 0 & 0\end{pNiceArray}, \qquad Z_n :=     
	\begin{pNiceArray}{c c c }
	-\lambda_n & 0 &0 \\
	0 &-\lambda_n & \lambda_n \\ 
	\lambda_n & 0 & -\lambda_n\end{pNiceArray}=P_2\diag(E_1\l_n)Q_2^{\top},
	\end{align*}
 }
 and $P_2, \, Q_2$ are the matrices from \eqref{Eq::P1P2_nsite} and \eqref{Eq::Q1Q2_nsite}. Therefore, for $n=1$, we have:
\begin{align*}
	&N'_1 \diag(E_1 \lambda ) A^{'\top}_1 =P_2\diag(E_1\l_n)Q_2^{\top}=  C_1.
\end{align*}
  To prove the result we show that $C_n$ is obtained by performing elementary row operations on $N_n'\diag(E_n\lambda)A_n^{'\top}$. Consider the matrix $X_n \in \mathbb{R}^{3n \times 3n}$ given by
  \begin{align*} (X_n)_{ij} := \begin{cases}
  1 ~~\text{if}~~ i=j\\
  1 ~~\text{if}~~ k<n, ~~ i=3k, ~~ j>3k\\
  0 ~~\text{else}.~~
  \end{cases} 
  \end{align*}
  Furthermore, we set $X_1:=\Id_3$. Multiplying by $X_n$ is the same as adding the rows $3k + 1, 3k + 2,$ and $3k+3$ to the row $3k $ for each $k = 1,\dots,n-1$. Note that $X_n$ can be written recursively as
  \[X_n =\begin{pNiceArray}{c c}
  X_{n-1}&X'_{n-1}\\ 
  0_{3\times (3n-3)} & \Id_3
  \end{pNiceArray}\] 
  where $X'_n\in \RR^{(3n-3)\times 3}$ and 
  \begin{align*} (X'_{n-1})_{ij} := \begin{cases}
  1 ~~\text{if}~~ i=3k, \,k=1,\ldots,n-1,\\
  0 ~~\text{else}.~~
  \end{cases} 
  \end{align*}
  It follows that for all $n > 1$ the matrix $X_n N'_n$ can be written recursively as:
\begin{align*} X_n N'_n =  \begin{pNiceArray}{c c }
	X_{n-1}N'_{n-1}& X_{n-1}P_1'+X'_{n-1}P_2\\ 
	0_{3\times (4n-4)}& P_2\\
	\end{pNiceArray}=\begin{pNiceArray}{c c }
	X_{n-1}N'_{n-1}& 0_{(3n-3)\times 4}\\ 
	0_{3\times (4n-4)}& P_2\\
	\end{pNiceArray},\end{align*}
	where $P_1'$ is obtained by removing the first three rows of $P_1$ and $P_2$ is defined as in \eqref{Eq::P1P2_nsite}.

 Finally, we show inductively that $C_n = X_n N'_n \diag(E_n \lambda)A^{'\top}_n$ for all $ n \in \mathbb{N}$. The statement holds for $n=1$. Assume that for $ n \in \mathbb{N}$, $C_k = X_k N'_k \diag(E_k \lambda )A^{' \top}_k$ holds for all $ 2\leq k < n$. Using the block structure of $X_n N'_n$, $E_n \lambda$, $A^{'\top}_n$ and $C_n$, we have:
	\begin{align*}
	X_{n}N'_{n}\diag(E_{n}\lambda)A^{'\top}_{n} &= \begin{pNiceArray}{c c }
	X_{n-1}N'_{n-1}\diag(E_{n-1}\hat{\l})& 0_{(3n-3)\times 4}\\ 
	0_{3\times (3n-3)}& P_2\diag(E_1\lambda_{n})\\
	\end{pNiceArray}  \begin{pNiceArray}{c c }
	A_{n-1}^{'\top}&0_{(3n-3)\times 3}\\
	Q_1^{'\top}&Q_2^{\top}\\
	\end{pNiceArray} \\
	&=\begin{pNiceArray}{c c c| c}
	&C_{n-1}&&0_{3\times 3}\\ \hdottedline
	&\phantom{-}0_{3\times (3n-6)}&V_n&Z_n\end{pNiceArray}=C_{n},
	\end{align*}
	where $\l = (\l_1, \dots , \l_n)$, $\hat{\l} = (\l_1, \dots ,\l_{n-1})$, and $Q_1,Q_2$ are the matrices from \eqref{Eq::Q1Q2_nsite}, with $Q_1'$ being obtained from $Q_1$ by removing the first three rows.
	
	Since $\det(X_n)=1$, we get
	\[ \det( N'_{n}\diag(E_{n}\lambda)A^{'\top}_{n} ) = \det( C_n) = \prod_{i=1}^n \det(Z_i) = (-1)^n \prod_{i=1}^n \lambda_i^3.\]
 Since this determinant is non-zero, the matrix $N_n'\diag(E_n\lambda)A_n^{\top}$ has rank $s = 3n$ for all $n$, and hence, its kernel has dimension $3$.
\end{proof}
 
Next, we find an explicit representation of a Gale dual matrix of $N_n'\diag(E_n\lambda)A_n^{\top}$ for all $n$.

\begin{theorem}
	\label{Prop::D_nsite}
	Let $N'_n,A_n$ and $E_n$ be the matrices associated with the network \eqref{Eq::nsitenetwork_reduced} and $\l= (\l_1, \dots , \l_n)$. Consider the matrices
 \[D^{(0)} := \begin{pNiceArray}{ c c c}[margin]
	1 & 1 &-1 \\
	0 & 1 & 0 \\ 
	0 & 0 & 1 
	\end{pNiceArray}, \qquad D^{(i)} := \begin{pNiceArray}{ c c c}[margin]
	0 & 0 &-1 \\
	-(i-1) & -(i-1) & -i \\ 
	1 & 1 & 1 
	\end{pNiceArray}, \quad i = 1, \dots , n, \quad \text{and}\]
	\begin{align*}
	D_n^{\top}(\lambda) := \begin{pmatrix} D^{(0)} &\dots & D^{(n)}\end{pmatrix} \in \mathbb{R}(\lambda)^{3 \times (3n+3)}.
	\end{align*}
	 Then $D_n(\lambda)$ is a Gale dual matrix of $N_n'\diag(E_n\lambda)A_n^{\top} \in \mathbb{R}(\lambda)^{3n \times (3n+3)}$.
\end{theorem}
\begin{proof}
By Lemma \ref{lem::detCn}, $\ker(N_n'\diag(E_n\l)A_n^{\top})$ has dimension 3. The columns of $D_n(\lambda)$ are linearly independent and hence, it is enough to show that they are contained in the kernel of $N_n'\diag(E_n\lambda)A_n^{\top}$. 
Note that a vector $v \in \mathbb{R}(\lambda)^{3n+3}$ is in this kernel if and only if $\diag(E_n\lambda)A_n^{\top}v \in \ker(N_n')$. Moreover, $ \ker(N_n')= \ker (N_n)$. Since the columns of $E_n$ form a basis of $\ker(N_n)$ by Proposition~\ref{Lemma::ExVectorsReduced_nsite}, it follows that $v \in \ker(N_n'\diag(E_n\lambda)A_n^{\top})$ if and only if there exists $\mu \in \mathbb{R}(\lambda)^{n}$ such that 
	\[ \diag(E_n \lambda) A_n^{\top}v = E_n \mu = \begin{bmatrix} u_1  & \ldots & u_n  \end{bmatrix}^{\top}\]
where 
\[u_i=\begin{bmatrix} \mu_i & \mu_i & \mu_i & \mu_i  \end{bmatrix}\in \mathbb{R}(\lambda)^{4} \qquad \text{for } i=1,\ldots, n.\] 
	Using the block form of $A_n$ and the block form of $E_n$, for each $v \in \mathbb{R}(\lambda)^{3n+3}$ we have 
	\begin{align}
	\label{Eq::Proof::Prop::D_nsite}
	\diag(E_n\lambda) A_n^{\top}v = \begin{bmatrix} \omega_1  &  \ldots & \omega_n \end{bmatrix}^{\top}.
	\end{align}
 such that
 \begin{align*}
	\omega_i= \begin{bmatrix} \lambda_i (v_1 + v_{3i}) & \lambda_i v_{3i+1} &  \lambda_i (v_2 + v_{3i+3}) & \lambda_i v_{3i+2}\end{bmatrix} \qquad \text{for } i=1,\ldots n.
	\end{align*}
So, the vector $v\in \RR(\l)$ is such that the entries $4i+1,4i+2,4i+3,$ and $4i+4$ are equal in the column vector in \eqref{Eq::Proof::Prop::D_nsite} for all $i=1,\ldots, n.$ In particular, $v_1+v_{3i}=v_{3i+1}=v_2+v_{3i+3}=v_{3i+2}.$ It is easy to check that the rows of $D_n^{\top}(\l)$ are linearly independent and satisfy these relations. This concludes the proof.
 \end{proof}

	Note that the Gale dual matrix obtained in Theorem~\ref{Prop::D_nsite} is independent of $\l$. While this is the case for \eqref{Eq::nsitenetwork_reduced}, we will see in Section~\ref{Section::Gunawardena} that the same does not hold for \eqref{Eq::Gunawardena_reduced}.

The main objective in this section is to obtain the expression of critical polynomial using Theorem~\ref{Prop:CritPolyWnDn}. The final missing ingredient is the expression of $\delta_n(\l)$. We will use \eqref{Eq::Delta} to find this expression. 
 
\begin{proposition}\label{prop:delta_nsite_gd}
    With the choice of Gale dual matrix in Theorem~\ref{Prop::D_nsite}, it holds
	\[ \delta_n(\lambda) =  \prod_{i=1}^n \lambda_i^3.\]
\end{proposition}

\begin{proof}
Consider the index set $I=\{4, \ldots, 3n+3\}$. Then $\tau_I$ is the permutation that sends $(1,\ldots,3n+3)$ to $(4,\dots,3n+3,1,2,3)$. A straightforward computation gives
\[\sgn ( \tau_I ) =  (-1)^n.\]
Moreover, $\det\big( [D_n^{\top}(\lambda)]_{[3],I^c} \big)=\det D^{(0)}=1.$ 
Using Lemma~\ref{lem::detCn} and substituting these values in \eqref{Eq::Delta}, we obtain the statement.
 \end{proof}   
	
\subsection{Connectivity of the Multistationarity Region of \eqref{Eq::nsitenetwork}}
\label{Section::Proofnsite}
 For $h=(h_1,\ldots,h_{3n+3})\in\RR^{3n+3}_{>0}$, by Theorem \ref{Prop:CritPolyWnDn} and Theorem \ref{Prop::D_nsite}, the critical polynomial associated with \eqref{Eq::nsitenetwork_reduced} can be written as $q_n(h,\lambda) =  g_n(h)\delta_n(\lambda)$, where 
\begin{align}
\label{Eq::gn}
g_n(h) = \sum_{ \substack{I \subseteq [3n+3] \\ |I| = 3}} \det \big( [W_n]_{[3],I} \big) \det \big( [D_n^{\top}]_{[3],I} \big) \prod_{i\in I^c}h_i \quad \text{ and } \quad \delta_n(\lambda) = \prod_{i=1}^n \lambda_i^3.
\end{align}

\begin{remark}\label{rem::gn_nsite}
    The polynomial $g_n$ is independent of $\l.$ Moreover, $\delta_n(\l)$ is a positive function if $\l\in\RR^{n}_{>0}$ and hence, if $g_n^{-1}(\mathbb{R}_{<0})$ is connected, then so is $q_n^{-1}(\mathbb{R}_{<0})$. Therefore, it is enough to consider the polynomial $g_n(h).$
\end{remark} 

Henceforth, we will use the following notation:
\begin{align}\label{eq::alpha}\alpha_{n,I}:=\det\big( [W_n]_{[3],I} \big) \det \big( [D_n^{\top}]_{[3],I} \big).\end{align} In the next remark, we list various cases when $\alpha_{n,I}=0.$

\begin{remark}\label{rem:casesindexing}
	For $n  \geq 1$ and $I\subset [3n+3]$ with $\lvert I \rvert = 3$, the following holds:
	\begin{enumerate}
		\item[(i)] If $3\ell + k, \, 3\ell'+k \in I$ for $\ell,\ell' \in [n]$ and $ k \in \{1,2\}$, then $\W=0.$ 
            \item[(ii)] If $3\ell + 3, \, 3\ell'+3 \in I$ for $\ell,\ell' \in [n] \cup \{0\}$, then $\W=0.$ 
            \item[(iii)] If $3\ell+1 \notin I$ or $3\ell'+2 \notin I$ for $\ell , \ell' \in [n]\cup \{0\}$, then $\W=0.$
		\item[(iv)] If $3\ell+1,3\ell+2 \in I$ for $\ell\in[n]$, then $\det \big( [D_n^{\top}]_{[3],I})=0.$
	\end{enumerate} 
\end{remark}

To show that the region in the parameter space for \eqref{Eq::nsitenetwork} and \eqref{Eq::nsitenetwork_reduced} that enables multistationarity is path connected for all $n$, we first write $g_n$ as a polynomial in $h_{3n+1},h_{3n+2},h_{3n+3}$ with coefficients in $\hat{h} = (h_1,\ldots,h_{3n})$:

\begin{equation}
\label{Eq::polygn_decomp}
\begin{aligned}
g_n(h) =&  \sum_{ J \subseteq [3]} \Big( a_{J}(\hat{h}) \prod_{i \in J^c} h_{3n+i} \Big) =~~b_0(h)+b_1(h), \notag
\end{aligned}
\end{equation}
where
\begin{align}\label{Eq::gnsplit}
\begin{split}
b_0(h)=&~~a_{\{1\}}(\hat{h})h_{3n+2}h_{3n+3} + a_{\{1,2\}}(\hat{h}) h_{3n+3} + a_{\{1,3\}}(\hat{h}) h_{3n+2} + a_{\{1,2,3\}}(\hat{h}) \\ 
b_1(h)=&~~(a_\emptyset(\hat{h})h_{3n+2}h_{3n+3}   + a_{\{2\}}(\hat{h}) h_{3n+3} + a_{\{3\}}(\hat{h}) h_{3n+2}+a_{\{2,3\}}(\hat{h}))h_{3n+1} .
\end{split}
\end{align}

In the next result, we focus on the polynomial $b_0(h)$.

\begin{proposition}\label{prop::b2positive}
	For all $n \geq 1$, the polynomial $b_0(h)$ is non-zero and all of its coefficients are positive.
\end{proposition}
\begin{proof}
	First we show that $a_{\{1,2\}}(\hat{h})$ and $a_{\{1,2,3\}}(\hat{h})$ are zero polynomials. To see this, note that their coefficients are obtained from summands in \eqref{Eq::gn} indexed by $I$ such that $\{3n+1,3n+2\} \subseteq I$. For such an $I$, we have $\det\big( [D_n^\top]_{[3],I} \big) = 0$ by Remark~\ref{rem:casesindexing}(iv). Thus, $a_{\{1,2\}}(\hat{h}) = 0$ and $a_{\{1,2,3\}}(\hat{h})=0$.

	We now show that $a_{\{1,3\}}(\hat{h})$ and $a_{\{1\}}(\hat{h})$ only have positive coefficients. First consider the polynomial $a_{\{1,3\}}(\hat{h})$. Here the coefficients \eqref{eq::alpha} of $a_{\{1,3\}}(\hat{h})$ are computed for $I = \{r,3n+1,3n+3\}$ such that $r\in [3n]$. By Remark~\ref{rem:casesindexing}, we only need to consider the case when $r=3\ell+2$ for $\ell=0,\ldots,n-1.$ When $\ell=0$, we have that $\det\big( [D_n^\top]_{[3],I} \big)=0$.
	When $\ell \in [n-1] $, we have:
	\[\alpha_{n,I} = \det \begin{pmatrix} 0 & 1 & 0 \\ 1 & 0 & 0 \\ 1 & 1 & 1\end{pmatrix}  \det \begin{pmatrix} 0 & 0 & -1 \\ -(\ell-1) & -(n-1) & -n \\ 1 & 1 & \phantom{-}1\end{pmatrix}= (-1)(\ell-n) = n-\ell >0.\] 
	This ensures $a_{\{1,3\}}(\hat{h})$ only has positive coefficients.
	
	We now consider the coefficients of $a_{\{1\}}(\hat{h})$. In this case, $3n+1 \in I$ and $3n+2,3n+3 \notin I$. We can write $I=\{r_1,r_2,3n+1\}$ for $r_1,r_2\in [3n]$ and let $\ell,\ell'\in [n-1].$ 
	\begin{itemize}
		\item[(i)] For $(r_1,r_2)=(1,2)$ and $(r_1,r_2)=(2,3)$, we have $\alpha_{n,I}=1$ and $\alpha_{n,I}=n$ respectively.
		\item[(ii)] Finally, $(r_1,r_2)\in\{(1,3\ell+2),(2,3\ell+2),(2,3\ell+3),(3,3\ell+2),(3\ell+2,3\ell'+3)\}$, we get, $\alpha_{n,I}=n-\ell.$
	\end{itemize}
	
	 By Remark~\ref{rem:casesindexing}, in all the other cases $\alpha_{n,I}=0.$ Since in all the cases above $\alpha_{n,I}\geq 0$, the coefficients of $a_{\{1\}}(\hat{h})$ are positive and hence, the result.
\end{proof}

Following corollary is now a direct consequence of Proposition~\ref{prop::b2positive}.

\begin{corollary}\label{Cor::b0positive} For $n \geq 2$, the face $\N_{e_{3n+1}}(g_n)$ is a proper face of $\N(g_n)$ and $\sigma_-(g_n) \subseteq $ $\N_{e_{3n+1}}(g_n)$. In particular, $g_n$ satisfies the closure property and if the polynomial
	\[ b_1(h)= \big(a_\emptyset(\hat{h})h_{3n+2}h_{3n+3}   + a_{\{2\}}(\hat{h}) h_{3n+3} + a_{\{3\}}(\hat{h}) h_{3n+2}+a_{\{2,3\}}(\hat{h})\big)h_{3n+1} \]
	has one negative connected component, then so does $g_n$.
\end{corollary}

\begin{proof}
For $I_1 = \{1,2,3\}$, $\alpha_{n,I_1} = 1$, and therefore $z_{I_1} \in \sigma(g_n)$. From Proposition \ref{prop::b2positive}, there exists $I_2 \subseteq [3n+3]$ with $\lvert I_2 \rvert = 3$, $3n+1 \notin I_2$ and $z_{I_2} \in \sigma(g_n)$. Since $e_{3n+1} \cdot \mu \in \{0,1\}$ for all $\mu \in \sigma(g_n)$, we have that $\N_{e_{3n+1}}(g_n)$ is a proper face of $\N(g_n)$. 

Since $\sigma_-(g_n) \subseteq $ $\N_{e_{3n+1}}(g_n)$ by Proposition \ref{prop::b2positive}, the second part of the corollary follows from Theorem \ref{Thm::NegativeFace}.
\end{proof}

By Corollary~\ref{Cor::b0positive}, it is enough to focus on $b_1(h)$. Next, we consider the polynomial $a_\emptyset(\hat{h})$. We recall that $z_{ \{i,j,k\}} \in \{0,1\}^{3n+3}$ denotes the vector whose entries indexed by $i,j,k \in [3n+3]$ are zero and all the other entries are $1$. If $i,j,k \in [3n]$, we write $\hat{z}_{\{i,j,k\}} \in \mathbb{R}^{3n}$ for the vector that is obtained from $z_{\{i,j,k\}}$ by removing its last three coordinates.

\begin{proposition}\label{Lem::a0}
For $n \geq 2$, and  $J = \emptyset$, $a_J(\hat{h})$ is the polynomial $g_{n-1}(\hat{h})$ associated with the network~$(\mathcal{F}_{1,n-1})$. Furthermore, if $\hat{z}_{\{i,j,k\}} \in \sigma_-(g_{n-1})$, then $z_{\{i,j,k\}} \in \sigma_-(g_{n})$.
\end{proposition}
\begin{proof}
  The terms in $a_\emptyset(\hat{h})$ are computed from the summands indexed by $I$ in \eqref{Eq::gn} such that $ \{3n+1,3n+2,3n+3\} \cap I = \emptyset$. Note the following set equaltiy:
	\[ \big\{ I \subseteq [3n+3] \mid |I| = 3 \text{ and } \{3n+1,3n+2,3n+3\} \cap I = \emptyset \big\} = \big\{ I \subseteq [3n] \mid |I| = 3 \big\}.\]
	Using the block structures of $W_n$ and $D_n$, the first part of the statement follows directly from \eqref{Eq::gn}.

  For the second part, let $I$ be in the set above. Since $a_\emptyset(\hat{h}) = g_{n-1}(\hat{h})$,  by \eqref{Eq::gn}, \eqref{Eq::polygn_decomp}, the coefficient of $g_n$ corresponding to $z_{I}$ is the same as the coefficient of $g_{n-1}$ that corresponds to $\hat{z}_{I}$.
\end{proof}

Next, we look at the support of $a_{\{2,3\}}(\hat{h})h_{3n+1}$, \; $a_{\{2\}}(\hat{h})h_{3n+1}h_{3n+3}$ and $a_{\{3\}}(\hat{h})h_{3n+1}h_{3n+2}$. In each of these polynomials, every exponent vector is of the form $z_{ \{i,j,k\}}$ by \eqref{Eq::gn} and is a vertex of the Newton polytope.

\begin{proposition}\label{Prop::a1a5}
 For $n \geq 2$, consider the supports $\sigma(a_{\{2,3\}}(\hat{h})h_{3n+1})$ and $\sigma(a_{\{2\}}(\hat{h})h_{3n+1}h_{3n+3})$. Each set has exactly one positive exponent vector given by $z_{\{1,3n+2,3n+3\}}$ and $z_{\{1,2,3n+2\}}$, respectively. Moreover, both supports have a strict separating hyperplane, and it holds that
 \begin{itemize}
 	\item[(i)] $\sigma_-(a_{\{2,3\}}(\hat{h})h_{3n+1})= \big\{ z_{\{3\ell +1,3n+2,3n+3\}} \mid \ell \in [n-1] \big\},$
 	\item[(ii)] $ z_{\{1,3,3n+2\}},z_{\{1,3\ell+1,3n+2\}} \in \sigma_-\big(a_{\{2\}}(\hat{h})h_{3n+1}h_{3n+3}\big)$ for $\ell \in [n-1]$. 
 \end{itemize}
\end{proposition}

\begin{proof}
	Let us consider the polynomial $a_{\{2,3\}}(\hat{h})h_{3n+1}$. The terms in $a_{\{2,3\}}$ are obtained from the summands in \eqref{Eq::gn} indexed by $I= \{r,3n+2,3n+3\}$ where $r\in [3n].$ By Remark~\ref{rem:casesindexing}, we only need to consider the case when $r=3\ell+1$ for $\ell=0,\ldots,n-1.$ When $\ell=0,$ we have $\alpha_{n,I}=1>0.$
When $r = 3\ell + 1$, for $\ell = [n-1]$ , we have:
	\[ \alpha_{n,I} = \det \begin{pmatrix} 1 & 0 & 0 \\ 0 & 1 & 0 \\ 1 & 1 & 1\end{pmatrix}  \det \begin{pmatrix} 0 & 0 & -1 \\ -(\ell-1) & -(n-1) & -n \\ 1 & 1 & \phantom{-}1\end{pmatrix} = \ell - n < 0.\]
	This shows that $\sigma_-(a_{\{2,3\}}(\hat{h})h_{3n+1})= \big\{ z_{\{3\ell +1,3n+2,3n+3\}} \mid \ell \in [ n-1 ] \big\}.$  Additionally, for $I=\{1,3n+2,3n+3\}$ we get the unique positive exponent vector.
	
	For $a_{\{2\}}(\hat{h})h_{3n+1}h_{3n+3}$, we haven $I = \{r_1,r_2,3n+2\}$ for some $r_1,r_2 \in [3n]$. Let $\ell,\ell'\in[n-1]$. 
 \begin{itemize}
 \item[(i)] For $(r_1,r_2)=(1,2)$ we obtain $\alpha_{n,I}=1$.
  \item[(ii)] For $(r_1,r_2)\in\{(1,3),(1,3\ell+1)\}$, we get, $\alpha_{n,I}=-(n-1).$
  \item[(iii)] For $(r_1,r_2)\in\{(2,3\ell +1),(3,3\ell+1),(3\ell+1,3\ell'+3)\}$, we have $\alpha_{n,I}=-(n-\ell)$.
  \item[(iv)] Finally, for $(r_1,r_2)=(1,3\ell+3)$, we have $\alpha_{n,I}=-(n-\ell-1).$
  \end{itemize}
  By Remark~\ref{rem:casesindexing}, in all the other cases $\alpha_{n,I}=0.$
	From the computations above we obtain the unique positive exponent vector for $I=\{1,2,3n+2\}.$

To see that both $\sigma(a_{\{2,3\}}(\hat{h})h_{3n+1})$ and $\sigma(a_{\{2\}}(\hat{h})h_{3n+1}h_{3n+3})$ have a strict separating hyperplane, note that the unique positive exponent vector is a vertex of the coppesponding Newton polytope. Therefore, the unique positive exponent vector can be separated by an affine hyperplane from the other exponent vectors. This concludes the proof.
\end{proof}

Finally, we consider the set of exponent vectors of $a_{\{3\}}(\hat{h})h_{3n+1}h_{3n+2}$.

\begin{proposition}\label{Prop::a6}
	For $n \geq 3$, the set of exponent vectors $\sigma\Big(a_{\{3\}}(\hat{h})h_{3n+1}h_{3n+2}\Big)$ has a strict separating hyperplane. Moreover, we have
	$\big\{ z_{\{2, 3\ell +1,3n+3\}} \mid \ell \in [n-1] \big\} \subseteq\sigma_-\big(a_{\{3\}}(\hat{h})h_{3n+1}h_{3n+2}\big)$.
\end{proposition}
\begin{proof}
We consider $I=\{r_1,r_2,3n+3\}$ where $r_1,r_2\in [3n]$. First, we investigate the case $r_1=1, \, r_2 \neq 2$. By Remark~\ref{rem:casesindexing}, we only need to consider the case when $r_2=3\ell+2$ for $\ell\in[n-1].$ In this case $\alpha_{n,I}=(n-\ell+1)\geq0$ with equality only if $\ell = n-1$.

Next, we work with the case when $r_1=2, \, r_2 \neq 1$. Again by Remark~\ref{rem:casesindexing}, we only need to consider the case when $r_2=3\ell+1$ for $\ell\in[n-1].$ For this case we find $\alpha_{n,I}=-(n-\ell)<0.$
	This shows that the exponent vectors of the form $z_{\{2,3\ell+1,3n+3\}}$ correspond to negative coefficients. 
 
Furthermore,
 \begin{itemize}
     \item[(i)] When $1 \in I$ and $2 \notin I$, we have $(e_1 - e_2) \cdot z_I = -1<0$ .
     \item[(ii)] When either $1,2 \in I$ or $1,2 \notin I$, we have $(e_1 - e_2) \cdot z_I = 0.$
     \item[(iii)] Finally, when $2 \in I$ and $1 \notin I$, we have $(e_1 - e_2) \cdot z_I = 1>0$.
 \end{itemize}

	Thus, the open half-space $\mathcal{H}^{+,\circ}_{e_1-e_2,0}$ contains only negative exponent vectors and $\mathcal{H}^{-,\circ}_{e_1-e_2,0}$ contains only positive exponent vectors. This implies that $\mathcal{H}_{e_1-e_2,0}$ is a strict separating hyperplane of $\sigma\big(a_{\{3\}}(\hat{h})h_{3n+1}h_{3n+2}\big)$.
\end{proof}

\begin{proposition}
\label{Prop_Strongly_2site}
   The polynomial $g_2$ from \eqref{Eq::gn} has one negative connected component.
\end{proposition}

\begin{proof}
    The polynomial $g_2$ equals
\begin{align*}
    &-h_1 h_2 h_3 h_5 h_6 h_7 + h_2 h_3 h_4 h_5 h_6 h_7 + h_1 h_2 h_3 h_4 h_6 h_8 + 2 \, h_2 h_3 h_4 h_6 h_7 h_8 - h_1 h_3 h_5 h_6 h_7 h_8 \\
    &+h_3 h_4 h_5 h_6 h_7 h_8 - h_1 h_2 h_3 h_5 h_7 h_9 - h_1 h_2 h_5 h_6 h_7 h_9- h_1 h_3 h_5 h_6 h_7 h_9 - h_2 h_3 h_5 h_6 h_7h_9  \\
    &- h_2 h_4 h_5 h_6 h_7 h_9 + h_3 h_4 h_5 h_6 h_7 h_9 + h_1 h_2 h_3 h_4 h_8 h_9 + h_1 h_3 h_4 h_5 h_8 h_9 + h_1 h_2 h_4 h_6 h_8 h_9 \\
    &+ h_1 h_3 h_4 h_6 h_8 h_9 + h_2 h_3 h_4 h_6 h_8 h_9 + 2 \, h_1 h_4 h_5 h_6 h_8 h_9 +h_3 h_4 h_5 h_6 h_8 h_9 + h_2 h_3 h_4 h_7 h_8 h_9 \\
    &+ h_3 h_4 h_5 h_7 h_8 h_9 + h_3 h_4 h_6 h_7 h_8 h_9 + h_1 h_5 h_6 h_7 h_8 h_9 + h_3 h_5 h_6 h_7 h_8 h_9 + h_4 h_5 h_6 h_7 h_8 h_9.
\end{align*}
   An easy computation shows that the hyperplane $\mathcal{H}_{v,4}$ with $v = (1,1,1,0,1,0,1,0,1)$ is a strict separating hyperplane of $\sigma(g_2)$. The statement now follows by Proposition~\ref{Thm::SepHypThm}.
\end{proof}

We are now ready to show that the multistationarity region of \eqref{Eq::nsitenetwork_reduced} and  \eqref{Eq::nsitenetwork} are connected for all $n\geq 2$.
 
\begin{theorem} \label{Thm::nsite_reduced_connected}
For all $n \geq 2$, the parameter region of multistationarity is path connected for the networks \eqref{Eq::nsitenetwork_reduced} and \eqref{Eq::nsitenetwork}.
\end{theorem}
\begin{proof}
    By Corollary~\ref{Cor::b0positive}, it is enough to show that the following polynomial
\begin{align*}
\begin{split}
b_1 = \big(a_\emptyset(\hat{h})h_{3n+2}h_{3n+3} + a_{\{2,3\}}(\hat{h})  + a_{\{2\}}(\hat{h}) h_{3n+3} + a_{\{3\}}(\hat{h}) h_{3n+2}\big)h_{3n+1}
\end{split}
\end{align*}
has exactly one negative connected component. To show that $b_1$ has one negative connected component for all $n\geq 2$ we will use an induction argument over $n$. From Proposition~\ref{Prop_Strongly_2site}, we know that $g_2$ has exactly one negative component, so we assume that it holds for all $g_k$ where $2\leq k\leq n-1.$

Now, we write $b_1 = b_2 + b_3$ based on the exponent of $h_{3n+3}$:
{\small
 \begin{align}\label{Eq::b2b3}
	b_2 := \big(a_\emptyset(\hat{h})h_{3n+1}h_{3n+2} +  a_{\{2\}}(\hat{h}) h_{3n+1}\big)h_{3n+3} \quad  \text{and} \quad b_3 := a_{\{2,3\}}(\hat{h}) h_{3n+1} + a_{\{3\}}(\hat{h}) h_{3n+1}h_{3n+2}.\end{align}
 }
Note that $\sigma(b_2)\subseteq \N_{e_{3n+3}}(b_1)$ and $\sigma(b_3) \subseteq \N_{-e_{3n+3}}(b_1)$. Moreover, these two faces are parallel and $\sigma(b_1)=\sigma(b_2)\cup\sigma(b_3)$. Since by Proposition \ref{Prop::a1a5} we have $z_{\{1,4,3n+2\}} \in \sigma_-(b_2)$ and $z_{\{4,3n+2,3n+3\}} \in \sigma_-(b_3)$, using Proposition \ref{Prop::NegEdge} we obtain a negative edge $\Conv( z_{\{1,4,3n+2\}},z_{\{4,3n+2,3n+3\}})$ of $\N(b_1)$. Therefore by Theorem \ref{Thm::ParallelFaces}, it is now enough to show that $b_2$ and $b_3$ have one negative connected component.

 We now consider $b_2$. The two summands in \eqref{Eq::b2b3} split $b_2$ in terms of the exponent of $h_{3n+2}$. 
 Similar to $b_1$, we have $\sigma(a_\emptyset(\hat{h})h_{3n+1}h_{3n+2}h_{3n+3})\subseteq \N_{e_{3n+2}}(b_2)$ and $\sigma(a_{\{2\}}(\hat{h})h_{3n+1}h_{3n+3})$ is contained in  $\N_{-e_{3n+2}}(b_2)$. Moreover, from Proposition \ref{Prop::NegEdge}, Proposition \ref{Lem::a0} and Proposition \ref{Prop::a1a5} it follows that $\Conv( z_{\{1,3,3(n-1)+2\}},z_{\{1,3,3n+2\}})$ is an edge of the Newton polytope $\N(b_2)$ between the two negative exponent vectors. Using Theorem \ref{Thm::ParallelFaces}, it is enough to show that the two summands have exactly one negative connected component.
 By Proposition~\ref{Lem::a0}, $a_\emptyset(\hat{h}) =  g_{n-1}(\hat{h})$, and it has one negative connected component by the inductive assumption. On the other hand, by Proposition~\ref{Prop::a1a5}  $a_{\{2\}}(\hat{h}) h_{3n+1}h_{3n+3}$ has a strict separating hyperplane, which implies that $ a_{\{2\}}(\hat{h}) h_{3n+1}h_{3n+3}$ has a unique negative connected component by Proposition~\ref{Thm::SepHypThm}. Hence, so does $b_2.$

To complete the proof, we show that $b_3$ has one negative connected component. We follow the same argument thread as for $b_2$. We split the two summands of $b_3$ in \eqref{Eq::b2b3} in terms of the exponent of $h_{3n+2}.$ Their exponent vectors then lie on parallel faces of $\N(b_3)$ with normal vectors $\pm e_{3n+2}$. We observe that  $\Conv( z_{\{4,3n+2,3n+3\}},z_{\{2,4,3n+3\}})$ is an edge of the Newton polytope $\N(b_3)$ between two negative exponent vectors of $a_{\{2,3\}}(\hat{h}) h_{3n+1}$ and $a_{\{3\}}(\hat{h}) h_{3n+1}h_{3n+2}$, respectively. By Theorem  \ref{Thm::ParallelFaces}, it is therefore enough to show that the two summands have exactly one negative connected component. By Proposition~\ref{Prop::a1a5} and Proposition \ref{Prop::a6} the supports of $a_{\{2,3\}}(\hat{h}) h_{3n+1}$ and $a_{\{3\}}(\hat{h}) h_{3n+1}h_{3n+2}$ have a strict separating hyperplane. Therefore, both $a_{\{2,3\}}(\hat{h}) h_{3n+1}$ and $ a_{\{3\}}(\hat{h}) h_{3n+1}h_{3n+2}$ have one negative connected component and hence, so does $b_1$. 
\end{proof}

\section{\bf Weakly Irreversible Phosphorylation Networks}
\label{Section::Weakly}

We now study the family of reduced weakly irreversible phosphorylation networks \eqref{Eq::Gunawardena_reduced}. This section has several results analogous to Section~\ref{Section::SIP}. First, in Section~\ref{Section::Gunawardena} we give the basic set up that allow us to compute the critical polynomial of the reduced networks. In Section~\ref{Section::ProofGunawardena} we use this to show that this network has multistationarity for all $n\geq 2$ and the parameter region of multistationarity is connected.

\subsection{Computation of Critical Polynomial}
\label{Section::Gunawardena}

The family of networks \eqref{Eq::Gunawardena_reduced} is obtained from \eqref{Eq::Gunawardena} by removing reversible reactions as in Theorem~\ref{Thm::ReducedConnectivity}. We fix the order of the species as $K,F,S_0,Y_1,Y_2,U_1,U_2,S_1, \dots, Y_{2n-1},Y_{2n},U_{2n-1},U_{2n},S_{n}$ to obtain the stoichiometric matrix $N_n$ and the reactant matrix $A_n$ as below. For $n=1$ these matrices are given by:
\begin{align} 
{ N_1 = \begin{pNiceMatrix}
  P_1\\ P_2 \end{pNiceMatrix}  }, \qquad \text{and} \qquad  {  A_1 = \begin{pNiceMatrix}Q_1\\Q_2 \end{pNiceMatrix}. }
\end{align}

For general $n\geq 2$, we can write these matrices recursively as:

\begin{align}
\label{Eq::NA_recursive_Gunawardena}
N_n =\begin{pNiceArray}{cc}
                          N_{n-1}&P_1\\
                          ~~0_{5 \times (8n-8)}&P_2
                          \end{pNiceArray}\in \RR^{(5n+3)\times 8n}, 
                          \quad 
                         A_n= \begin{pNiceArray}{cc}
                          A_{n-1}&Q_1\\
                          ~~0_{5 \times (8n-8)}&Q_2
                          \end{pNiceArray}\in \mathbb{R}^{(5n+3) \times 8n},
\end{align}             
where 
{\footnotesize  \begin{align}
   \label{Eq::P1P2_Gunawardena} 
  P_1 =\begin{pNiceArray}{cccccccc}
         -1 & 0 & 1 & -1 & \phantom{-}0 & 0 & 0 & \phantom{-}0 \\
\phantom{-}0 & 0 & 0 & \phantom{-}0 & -1 & 0 & 1 & -1 \\
\phantom{-}0 & 0 & 0 & \phantom{-}0 & \phantom{-}0 & 0 & 0 & \phantom{-}0 \\
\phantom{-}\vdots&\vdots&\vdots&\phantom{-}\vdots&\phantom{-}\vdots&\vdots&\vdots&\phantom{-}\vdots\\
\phantom{-}0 & 0 & 0 & \phantom{-}0 & \phantom{-}0 & 0 & 0 & \phantom{-}0 \\
-1 & 0 & 0 & \phantom{-}0 & \phantom{-}0 & 0 & 1 & -1 \\
     \end{pNiceArray} \in \RR^{(5n-2) \times 8} ,
      P_2 =\begin{pNiceArray}{cccccccc}
1 & -1 & \phantom{-}0 & \phantom{-}0 & \phantom{-}0 & \phantom{-}0 & \phantom{-}0 & 0 \\
0 & \phantom{-}1 & -1 & \phantom{-}1 & \phantom{-}0 & \phantom{-}0 & \phantom{-}0 & 0 \\
0 & \phantom{-}0 & \phantom{-}0 & \phantom{-}0 & \phantom{-}1 & -1 & \phantom{-}0 & 0 \\
0 & \phantom{-}0 & \phantom{-}0 & \phantom{-}0 & \phantom{-}0 & \phantom{-}1 & -1 & 1 \\
0 & \phantom{-}0 & \phantom{-}1 & -1 & -1 & \phantom{-}0 & \phantom{-}0 & 0
     \end{pNiceArray} \in \RR^{5 \times 8} ,
 \end{align} }

and 

 {\footnotesize
 \begin{align}
 \label{Eq::Q1Q2_Gunawardena}
  Q_1 =\begin{pNiceArray}{cccccccc}
   1 & 0 & 0 & 1 & 0 & 0 & 0 & 0 \\
0 & 0 & 0 & 0 & 1 & 0 & 0 & 1 \\
0 & 0 & 0 & 0 & 0 & 0 & 0 & 0 \\
\vdots&\vdots&\vdots&\vdots&\vdots&\vdots&\vdots&\vdots\\
0 & 0 & 0 & 0 & 0 & 0 & 0 & 0 \\
1 & 0 & 0 & 0 & 0 & 0 & 0 & 1 \\
     \end{pNiceArray}\in \RR^{(5n-2)\times 8} , \quad 
    Q_2 =\begin{pNiceArray}{cccccccc}
0 & 1 & 0 & 0 & 0 & 0 & 0 & 0 \\
0 & 0 & 1 & 0 & 0 & 0 & 0 & 0 \\
0 & 0 & 0 & 0 & 0 & 1 & 0 & 0 \\
0 & 0 & 0 & 0 & 0 & 0 & 1 & 0 \\
0 & 0 & 0 & 1 & 1 & 0 & 0 & 0
     \end{pNiceArray}\in \RR^{5\times 8}.
 \end{align}
}

\begin{remark}\label{rem::NnrankN2}
	The matrices $P_1$ and $Q_1$ have $5n-5$ rows with all entries zero. Note that the matrix $P_2$ has full rank and the rank of $N_1$ is 5. From the recursive relation it follows that $\rank(N_n)=5n.$ Equivalently, the dimension of left kernel of $N_n$ is 3 and $\dim(\im(N_n))=5n.$
\end{remark}

Using the above matrices, we compute a conservation law matrix and an extreme matrix for \eqref{Eq::Gunawardena_reduced} in Lemma~\ref{lem::Wn_Gunawardena} and Proposition~\ref{Lemma::ExVectorsReduced_Gunawardena} respectively.

\begin{lemma}
\label{lem::Wn_Gunawardena}
	Given the stoichiometric matrix $N_n$ in \eqref{Eq::NA_recursive_Gunawardena}, a conservation law matrix $W_n$ is given~by:
	\begin{align}\label{Eq::Wn_Gunawardena}
	W_n:=\begin{pmatrix}
	\Id_3&\smash[b]{\blockW{n}}
	\end{pmatrix}\in \RR^{3\times(5n+3)},
	\end{align}
	\medskip
 
\noindent	where $\Id_3$ is the identity matrix of size 3 and 
	\[\overline{W}:=\begin{pNiceArray}{ccccc}
	1&1&0&0&0\\
	0&0&1&1&0\\
	1&1&1 &1&1
	\end{pNiceArray}\in \RR^{3\times 5}.\]
\end{lemma}

\begin{proof}
Let $W_0:=\Id_3$, then using $W_n=\begin{pmatrix}
W_{n-1}&\overline{W}
\end{pmatrix}$, for $n\geq 2$ we have 
\[W_n N_n=\begin{pNiceArray}{cc}
	W_{n-1}N_{n-1}&W_{n-1}P_1+\overline{W}P_2\end{pNiceArray}.\]
Simple computation gives: 
{\small
\begin{align*}
W_{n-1}P_1 =\begin{pNiceArray}{cccccccc}
-1&0&1&-1&\phantom{-}0&0&0&\phantom{-}0\\
\phantom{-}0&0&0&\phantom{-}0&-1&0&1&-1\\
-1&0&0&\phantom{-}0&\phantom{-}0&0&1&-1
\end{pNiceArray} \quad \text{ and } \quad \overline{W}P_2 =\begin{pNiceArray}{cccccccc}
1&0&-1&1&0&0&\phantom{-}0&0\\
0&0&\phantom{-}0&0&1&0&-1&1\\
1&0&\phantom{-}0&0&0&0&-1&1
\end{pNiceArray}.
\end{align*}}
For $n=1$, we get $W_1 N_1=W_0P_1+\overline{W}P_2 = 0$. We now assume that the claim holds for all $k\in\mathbb{N}$ such that $1\leq k \leq n-1$. By assumption $W_{n-1}N_{n-1}$ is a zero matrix and
hence, $W_n N_n=0_{3\times 8n}$ and since the rank of $W_n$ and the dimension of the left kernel of $N_n$ are both $3$ for all $n$ by Remark~\ref{rem::NnrankN2}, it follows that $W_n$ is a conservation law matrix for $N_n$.
\end{proof}

 \begin{proposition}
\label{Lemma::ExVectorsReduced_Gunawardena}
An extreme matrix of \eqref{Eq::Gunawardena_reduced} has the following recursive form 
\begin{align}
\label{Eq::EnGunawardena}
E_n = 
\begin{pNiceArray}{c | c}
E_{n-1} & 0_{(8n-8)\times 3}\\ \hdottedline
~~0_{8\times (3n-3)}& E_1\end{pNiceArray} \in \mathbb{R}^{8n\times3n},
\end{align}
where
\[ E_1 =\begin{pNiceMatrix}
    1  &1  &1  &0  &1  &1  &1  &0    \\
     0  &0  &1 &1  &0  &0  &0  &0    \\
     0  &0  &0 &0  &0  &0  &1  &1    \\
\end{pNiceMatrix}^{\top} \in \mathbb{R}^{8 \times 3}.\]
Moreover, the column vectors $E_n^{(1)}, \dots , E_n^{(3n)} \in \mathbb{R}^{8n}$ of $E_n$ form a basis of $\ker(N_n).$
 \end{proposition}

\begin{proof}
Using \eqref{Eq::NA_recursive_Gunawardena}, a simple computation shows that  $E_n^{(1)}, \dots ,E_n^{(3n)} \in \ker(N_n) \cap \mathbb{R}^{8n}_{\geq 0}$. Since by Remark~\ref{rem::NnrankN2}
\[ \dim \ker(N_n) = 8n - \dim \im (N_n) = 8n - 5n = 3n\]
and $E_n^{(1)}, \dots ,E_n^{(3n)}$ are linearly independent, they form a basis of $\ker(N_n)$.

Furthermore, $E_1^{(1)},E_1^{(2)},E_1^{(3)}$ satisfy \eqref{Eq::Assumption::ExtremeVector} and since the blocks of $E_n$ have pairwise disjoint supports, $E_n^{(1)}, \dots ,E_n^{(3n)}$ also satisfy \eqref{Eq::Assumption::ExtremeVector}. Thus, $E_n$ is an extreme matrix for \eqref{Eq::Gunawardena_reduced} by Proposition~\ref{Prop::ExtremeVector}. 
\end{proof}

The matrix $N_n'$ is obtained by removing the first three rows of $N_n.$ 
We will now compute a Gale dual $D_n(\lambda)$ (Theorem~\ref{Prop::D_Gunawardena}) of $N_n'\diag(E_n\l)A_n^{\top}$ and $\delta_n(\lambda)$ (Proposition \ref{prop::delta_Gw}) for the network \eqref{Eq::Gunawardena_reduced}. 

In the following, we view  $\l =(\l_1,\ldots,\l_{3n})$ as symbolic variables. In Lemma~\ref{lem::detCn_Gw}, we first compute the maximal minor $\det( [ N_n'\diag(E_n\lambda)A_n^{\top} ]_{[s],I})$ for $I = \{4,\dots,3n+3\}$ and $I^c = \{1,2,3\}$. 
\begin{lemma}\label{lem::detCn_Gw}
  For a fixed $n$, let $s=5n$, $I=\{4,\ldots,5n+3\}$. Then, \[ \det([N_n'\diag(E_n\lambda)A_n^{\top}]_{[s],I})=(-1)^n \prod_{k=0}^{n-1} (\lambda_{3k+1} + \lambda_{3k+2})(\lambda_{3k+1} + \lambda_{3k+3}) \lambda_{3k+1}^3.\]
  In particular, $\ker(N_n'\diag(E_n\lambda)A_n^{\top})\in\RR(\l)^{5n}$ has dimension $3$.
\end{lemma}

\begin{proof}
    Let $A'_n$ be obtained from $A_n$ after removing the first three rows. Then we have, \[[N_n'\diag(E_n\lambda)A_n^{\top}]_{[s],I} = N_n'\diag(E_n\lambda)A_n^{'\top}.\]
  Consider the matrix $X_n \in \mathbb{R}^{5n \times 5n}$, $n\geq 2$ given by
    \begin{align*} (X_n)_{ij} := \begin{cases}
    1 ~~\text{if}~~ i=j,\\
    1 ~~\text{if}~~ k<n, ~~ i=5k, ~~ j>5k,\\
    0 ~~\text{else}.~~
    \end{cases} 
    \end{align*}
    Furthermore, we set $X_1:=\Id_5$. Multiplying by $X_n$ is the same as adding the rows $5k + 1, 5k + 2,5k+3,5k+4$ and $5k+5$ to the row $5k $ for each $k = 1,\dots,n-1$. The matrix $X_n$ can be written recursively as
    \[X_n=\begin{pNiceArray}{c c}
    X_{n-1}&X'_{n-1}\\ 
    0_{5\times (5n-5)} & \Id_5
    \end{pNiceArray} \qquad \text{with } \quad (X'_{n-1})_{ij} := \begin{cases}
    1 ~~\text{if}~~ i=5k, k=1,\ldots,n-1,\\
    0 ~~\text{else}.~~
    \end{cases} \] 
    
    It follows that for all $n > 1$ the matrix $X_n N'_n$ can be written recursively as:
    \begin{align*} X_n N'_n =  \begin{pNiceArray}{c c }
    X_{n-1}N'_{n-1}& X_{n-1}P_1'+X'_{n-1}P_2\\ 
    0_{5 \times (8n-8)}& P_2\\
    \end{pNiceArray}=\begin{pNiceArray}{c c }
    X_{n-1}N'_{n-1}& 0_{(5n-5)\times 8}\\ 
    0_{5\times (8n-8)}& P_2\\
    \end{pNiceArray},\end{align*}
    where $P_1$ and $P_2$ is defined as in \eqref{Eq::P1P2_Gunawardena} and $P_1'$ is obtained by removing the first three rows of $P_1$. We define the matrices $C_n \in \mathbb{R}(\lambda)^{5n \times 5n}$ recursively as:
\begin{align*}
 	C_n := 
 \begin{pNiceArray}{c c c| c}
 &C_{n-1}&&0_{5\times 5}\\ \hdottedline
 &0_{5\times (5n-10)}&V_n&Z_n\end{pNiceArray} \in \mathbb{R}(\lambda)^{5n\times5n} \quad \text{ and } \quad
 C_1:=Z_1
\end{align*}

where 
{\footnotesize

\begin{align*}
V_n &:= \begin{pNiceArray}{c c c c c}
0 & 0 & 0& 0& \lambda_{3n-2} \\
0 & 0 & 0  & 0 & 0 \\ 
0 & 0 & 0  & 0 & 0 \\ 
0 & 0 & 0  & 0 & \lambda_{3n} \\ 
0 & 0 & 0  & 0 & 0 \\ \end{pNiceArray}, \quad
Z_n &:=     
 \begin{pNiceArray}{c c c  c  c}
-\lambda_{3n-2} & 0 &0 &0 &0 \\
\lambda_{3n-2} &-\lambda_{3n-2}-\lambda_{3n-1} & 0 &0 &\lambda_{3n-1} \\ 
0 & 0& -\lambda_{3n-2} & 0 & \lambda_{3n-2}\\
0 & 0& \lambda_{3n-2} & -\lambda_{3n-2}-\lambda_{3n} & 0\\
0 & \lambda_{3n-2} + \lambda_{3n-1} & 0 & 0 & -\lambda_{3n-2} -\lambda_{3n-1}\end{pNiceArray}.
\end{align*}
}

Let $\omega = [\l_{3n-2},\l_{3n-1},\l_{3n}]^{\top}$. A simple computation shows that $Z_n=P_2\diag(E_1\omega)Q_2^{\top}$ and hence, for $n = 1$, we have $X_1N'_1 \diag(E_1 \lambda ) A^{'\top}_1  = C_1.$ 
Assume that $X_kN'_k \diag(E_k \lambda ) A^{'\top}_k=C_k$ for all $1\leq k< n$, and let $\hat{\l}=(\l_1,\dots,\l_{3n-3})$. Then using the block structure of $N_n$, $A_n$, and $E_n$, we get
\begin{align*}X_{n}N'_{n} \diag(E_{n} \lambda ) A^{'\top}_{n}&=\begin{pmatrix}
    X_{n-1}N_{n-1}'\diag(E_{n-1} \hat{\l})&0_{(5n-5)\times 8}\\
    0_{5 \times (8n-8)}&P_2 \diag(E_{1}\omega)
\end{pmatrix}\begin{pmatrix}
    A_{n-1}^{'\top}&0_{(8n-8)\times 5}\\
    Q_1^{'\top}&Q^{\top}_2
\end{pmatrix}\\&
= \begin{pNiceArray}{c c c| c}
&C_{n-1}&&0_{5\times 5}\\ \hdottedline
&\phantom{-}0_{5\times (5n-10)}&V_n&Z_n\end{pNiceArray}=C_{n},
 \end{align*}
 	where $Q_1,Q_2$ are the matrices from \eqref{Eq::Q1Q2_Gunawardena} and $Q_1'$ is obtained from $Q_1$ by removing the first three rows.
 Since $\det(X_n)=1$, it follows 
 \begin{align*} 
\det( N'_{n}\diag(E_{n}\lambda)A^{'\top}_{n} ) &= \det( C_n) =\prod_{k=1}^{n}\det(Z_k)  \\ &= (-1)^n \prod_{k=0}^{n-1} (\lambda_{3k+1} + \lambda_{3k+2})(\lambda_{3k+1} + \lambda_{3k+3}) \lambda_{3k+1}^3.\end{align*}
The second part of the lemma holds since the $5n \times 5n$ maximal minor $\det([N_n'\diag(E_n\lambda)A_n^{\top}]_{[s],I})$ is non-zero for all $n$.
\end{proof}

\begin{theorem}
\label{Prop::D_Gunawardena}
 Let $N_n,A_n$ and $E_n$ be the matrices associated with the network \eqref{Eq::Gunawardena_reduced} as described in \eqref{Eq::NA_recursive_Gunawardena} and \eqref{Eq::EnGunawardena}. Let $\l = (\l_1, \dots,\l_{3n})$ and for $k = 0, \ldots , n-1$ consider the matrices
  \begin{align*}
{\footnotesize 
 D^{(k)} := \begin{pNiceArray}{ c c c c c}[margin]
  0 & 0 & 0 & 0 & -1  \\
  -k & \tfrac{-k\lambda_{3k+1}-(k+1)\lambda_{3k+2}}{\lambda_{3k+1}+ \lambda_{3k+2}} & -k & \tfrac{-k\lambda_{3k+1}+(1-k)\lambda_{3k+3}}{\lambda_{3k+1}+ \lambda_{3k+3}} & -(k+1)  \\
  \phantom{-}1   & 1 & \phantom{-}1 & 1 & \phantom{-}1 
\end{pNiceArray} ,
 \quad
     D^{(-1)} := \begin{pNiceArray}{ c c  c}[margin]
  1 & 1 & -1  \\
  0 & 1 & \phantom{-}0    \\
  0 & 0 & \phantom{-}1
\end{pNiceArray}, \textrm{ and}}
\end{align*}
\begin{align*}
D_n^{\top}(\lambda) := \begin{pmatrix} D^{(-1)} &D^{(0)} &  \dots & D^{(n-1)}\end{pmatrix} \in \mathbb{R}(\l)^{3\times (5n+3)}.
 \end{align*}
 Then $D_n(\lambda)$ is a Gale dual matrix of $N_n'\diag(E_n \lambda )A_n^{\top}$.

\end{theorem}

\begin{proof}
A vector $v \in \mathbb{R}(\lambda)^{5n+3}$ lies in the kernel of $N_n'\diag(E_n\lambda)A_n^{\top}$ if and only if $\diag(E_n\lambda)A_n^{\top}v\in \ker(N_n').$ 
For $v \in \mathbb{R}(\lambda)^{5n+3}$, let 
\[\diag(E_n\lambda)A_n^{\top}v= \begin{bmatrix}
\omega_0,\ldots, \omega_{n-1}    
\end{bmatrix}^{\top}\]
such that for $i=0,\ldots, n-1$, $\omega_i \in \RR(\l)^8$ is given by:
{\small \begin{align} \label{Eq::Proof::Prop::D_Gunawardena}
   \Big[ \lambda_{3k+1}(v_1 + v_{5k+3}) \quad \lambda_{3k+1} v_{5k+4} \quad &(\lambda_{3k+1}+\lambda_{3k+2}) v_{5k+5}  \quad  \lambda_{3k+2}(v_1 + v_{5k+8}) \quad \notag \\
   &\lambda_{3k+1}(v_2 + v_{5k+8}) \quad \lambda_{3k+1} v_{5k+6}  \quad (\lambda_{3k+1}+\lambda_{3k+3}) v_{5k+7}  \quad \lambda_{3k+3}(v_2 + v_{5k+3})
\Big].\end{align}}
Moreover, since the columns of $E_n$ span $\ker(N_n'),$ there exists $\mu\in\RR(\l)^{3n}$ such that $\diag(E_n\lambda)A_n^{\top}v=E_n\mu.$ Let 
\[E_n\mu = \begin{bmatrix}
u_0,\ldots, u_{n-1}    
\end{bmatrix}^{\top}\]
such that for $i=0,\ldots,n-1$
\begin{align}\label{eq::uentries}u_i=\begin{bmatrix}
\mu_{3i+1}&\mu_{3i+1}&\mu_{3i+1}+\mu_{3i+2}&\mu_{3i+2}&\mu_{3i+1}&\mu_{3i+1}&\mu_{3i+1}+\mu_{3i+3}&\mu_{3i+3}
\end{bmatrix} \in \RR(\l)^8.\end{align}

Comparing equations \eqref{Eq::Proof::Prop::D_Gunawardena} and \eqref{eq::uentries}, we get that $v\in \ker(N_n'\diag(E_n\lambda)A_n^{\top})$ if it satisfies the following relations for all $i=0,\ldots,n-1$:
\begin{align*}
 &v_1+v_{5k+3}=v_{5k+4},\quad v_1+v_{5k+3}=v_2+v_{5k+8}, \quad v_1+v_{5k+3}=v_{5k+6},
\\&(\lambda_{3k+1}+\lambda_{3k+2}) v_{5k+5} =\lambda_{3k+1}(v_1 + v_{5k+3})+\lambda_{3k+2}(v_1 + v_{5k+8}),\\&(\lambda_{3k+1}+\lambda_{3k+3}) v_{5k+7}=\lambda_{3k+1}(v_1 + v_{5k+3})+\lambda_{3k+3}(v_2 + v_{5k+3}).\end{align*}

It is easy to check that the columns of $D_n(\l)$ are linearly independent and satisfy these relations. Moreover, since $\ker(N_n'\diag(E_n\l)A_n^{\top})$ has dimension 3 (Lemma \ref{lem::detCn_Gw}), $D_n(\l)$ is a Gale dual matrix of $\ker(N_n'\diag(E_n\l)A_n^{\top})$.\end{proof}

Using Lemma~\ref{lem::detCn_Gw}, we compute $\delta_n(\l)$, which is the final ingredient needed to calculate the critical polynomial of \eqref{Eq::Gunawardena_reduced} as in Theorem \ref{Prop:CritPolyWnDn}.

\begin{proposition}\label{prop::delta_Gw}
     With the choice of Gale dual matrix in Theorem~\ref{Prop::D_Gunawardena}, from \eqref{Eq::Delta} we get
 \[ \delta_n(\lambda) = \prod_{k=0}^{n-1} (\lambda_{3k+1} + \lambda_{3k+2})(\lambda_{3k+1} + \lambda_{3k+3}) \lambda_{3k+1}^3.\]
\end{proposition}
\begin{proof}
By Lemma~\ref{lem::generalGD}, $\delta_n(\l)$ in $\eqref{Eq::Delta}$ is independent of the choice of index set. For $I = \{4,\dots,5n+3\}$ we have
\begin{align*}
     \sgn(\tau_I) =(-1)^n \quad \text{ and } \quad \det([D_n^{\top}(\lambda)]_{[3],I^c}) = 1.
\end{align*}
Now, using Lemma~\ref{lem::detCn_Gw} and Lemma~\ref{Lemma::GaleDualMinors}, we obtain the result.
\end{proof}

\subsection{Connectivity of Multistationarity Region of \eqref{Eq::Gunawardena}}
\label{Section::ProofGunawardena}

Using the conservation law matrix $W_n$ \eqref{Eq::Wn_Gunawardena}, the Gale dual matrix $D_n(\l)$ of $N_n'\diag(E_n\l)A_n^{\top}$ from Theorem~\ref{Prop::D_Gunawardena}, and applying Theorem~\ref{Prop:CritPolyWnDn} we write the critical polynomial of \eqref{Eq::Gunawardena_reduced} as 
\begin{align}
\label{Eq::qnGunawardena}
 q_n(h,\lambda) =   \sum_{ \substack{I \subseteq [5n+3] \\ |I| = 3}} \delta_n(\lambda) \det \big( [W_n]_{[3],I} \big) \det \big( [D_n^{\top}(\lambda)]_{[3],I} \big) \prod_{i\in I^c}h_i .
 \end{align}

Unlike for  \eqref{Eq::nsitenetwork_reduced}, the matrix $D_n^{\top}(\l)$  depends on $\l$ in this case. Since the product 
\begin{align}
\label{Eq:alphanI_Gunawardena}
\alpha(\l)_{n,I}:=\det \big( [W_n]_{[3],I} \big) \det \big( [D_n^{\top}(\lambda)]_{[3],I} \big)
\end{align}
is not independent of $\l$, we cannot factor $q_n$ in terms of polynomial in $h$ and polynomial in $\l$ like we did in Section~\ref{Section::Proofnsite}. Here, we study the whole polynomial $q_n$. Going forward, for ease of notations, we will only write $\alpha_{n,I}$ for $\alpha(\l)_{n,I}$.

From $n-1$ to $n$, the critical polynomial depends on eight new variables $h_{5n-1}, \ldots , h_{5n+3}$, and $\lambda_{3n-2},\lambda_{3n-1},\lambda_{3n}$. For any subset $J \subseteq [5]$, we define
\begin{align}
    \label{Eq:DefofJhat}
\hat{J} := \{5n-2+j \mid j \in J \} \quad \text{and} \quad \hat{J}^c := \{5n-2+j \mid j \in [5] \setminus J \},
\end{align}
which will be used to index the variables $h_{5n-1}, \ldots , h_{5n+3}$. For instance, for $J = [5]$ we have $\{h_{5n-1}, \ldots , h_{5n+3} \}= \{ h_i \mid i \in \hat{J}\}$. Using this notation, we write $q_n$ as 
\begin{align}
\label{Eq::qnGunawardena_recursive}
q_n(h,\lambda) = \sum_{\substack{J \subseteq [5] \\ |J| \leq 3}} \Big( a_J(\hat{h},\lambda) \prod_{i \in \hat{J}^c}h_i \Big),
\end{align}
where $\hat{h} = (h_1, \dots , h_{5n-2})$. From \eqref{Eq::qnGunawardena}, it follows that the $a_J$'s have the following form
\begin{align}
\label{Eq::qnGunawardena_aJ}
 a_J(\hat{h},\lambda)  =   \sum_{ \substack{I \subseteq [5n+3], |I| = 3\\ \hat{J} \subseteq I, \, \hat{J}^c \cap I = \emptyset }} \delta_n(\lambda) \det \big( [W_n]_{[3],I} \big) \det \big( [D_n(\lambda)]_{I,[3]} \big) \prod_{i\in I^c \setminus \hat{J}^c}h_i.
 \end{align}
First, we relate $a_\emptyset$ to the critical polynomial of $(\mathcal{F}_{2,n-1})$. In what follows, we set $\hat{\lambda} = (\lambda_1, \dots , \lambda_{3n-3})$.
 
\begin{proposition}
\label{Lemma::Gunawardena_a0}
For $n \geq 2$ and $J = \emptyset$,  \[a_J(\hat{h},\lambda) = q_{n-1}(\hat{h},\hat{\lambda}) (\lambda_{3n-2} + \lambda_{3n-1})(\lambda_{3n-2} + \lambda_{3n}) \lambda_{3n-2}^3,\] 
where $q_{n-1}(\hat{h},\hat{\lambda})$ denotes the critical polynomial of $(\mathcal{F}_{2,n-1})$.
\end{proposition}

\begin{proof}
The summands in \eqref{Eq::qnGunawardena_aJ} depend on the subsets $I\subseteq [5n-2]$ with $\lvert I \rvert = 3$. Note that the following sets coincide
\[ \big\{ I \subseteq [5n+3] \mid |I| = 3 \text{ and } \{5n-1,5n,5n+1,5n+2,5n+3\} \cap I = \emptyset \big\} = \big\{ I \subseteq [5(n-1)+3] \mid |I| = 3  \big\}.\]
The proof now follows from \eqref{Eq::qnGunawardena} and \eqref{Eq::qnGunawardena_aJ} using the block structures of $W_n$ and $D_n(\lambda)$.
\end{proof}

In \eqref{Eq::qnGunawardena_aJ} the cardinality of $I$ is 3. Moreover, when $J \neq \emptyset$, the intersection $I\cap \{5n-1,5n,5n+1,5n+2,5n+3\}\neq \emptyset.$ In the next lemma, we list cases when $\alpha_{n,I}=0$ from \eqref{Eq:alphanI_Gunawardena}. The proof of Lemma~\ref{rem::zero-aJ} follows directly by considering the columns of the matrices $[W_n]_{[3],I}$ and $[D_n^{\top}(\lambda)]_{[3],I}$.

\begin{lemma}\label{rem::zero-aJ}
    Let $I\subset [5n+3]$ with $|I|=3$ and $\ell,\ell'\in [n-1]\cup \{0\}$.
    \begin{itemize}[leftmargin=0.75cm]
    	\item[(i)] If $5\ell+4,5\ell'+5\in I$, then $\W=0$ and hence, $\alpha_{n,I}=0.$
    	\item[(ii)] If $5\ell+6,5\ell'+7\in I$, then $\W=0$ and hence, $\alpha_{n,I}=0.$
    	\item[(iii)]  If $5\ell+4,5\ell+6\in I$, then $\det([D_n^{\top}(\lambda)]_{[3],I})=0$ and hence, $\alpha_{n,I}=0.$
     \item[(iv)]  If $5\ell+k,5\ell'+k\in I$ for $k\in \{4,5,6,7,8\}$, then $\W=0$ and hence, $\alpha_{n,I}=0.$
     \item[(v)] If $3,5\ell+8 \in I$,  then $\W=0$ and hence, $\alpha_{n,I}=0.$
     \item[(vi)] If $I \in \big\{ \{1,3,5\ell+4\},\{1,3,5\ell+5\},\{1,5\ell+4,5\ell'+8\},\{1,5\ell+5,5\ell'+8\}$,  then $\W=0$ and hence, $\alpha_{n,I}=0.$
       \item[(vii)] If $I \in \big\{ \{2,3,5\ell+6\},\{2,3,5\ell+7\},\{2,5\ell+6,5\ell'+8\},\{2,5\ell+7,5\ell'+8\}$,  then $\W=0$ and hence, $\alpha_{n,I}=0.$
    \end{itemize}
\end{lemma}

\begin{proposition}
\label{Prop::GunawardenaZero}For all $n \in \mathbb{N}$, if $\{1,2\} \subseteq J, \{3,4\} \subseteq J$ or $\{1,3\} \subseteq J$, then the polynomial $a_J$ is the zero polynomial. Equivalently, $a_{\{1,2\}},a_{\{1,3\}},a_{\{3,4\}},a_{\{1,2,3\}},a_{\{1,2,4\}},a_{\{1,2,5\}},a_{\{1,3,4\}},a_{\{1,3,5\}},a_{\{2,3,4\}},$ and $a_{\{3,4,5\}}$ are zero polynomials.
\end{proposition}

\begin{proof}
The result follows from Lemma~\ref{rem::zero-aJ}(i),(ii),(iii).
\end{proof}

To determine the signs of the coefficients of $q_n$ and to find vertices of the Newton polytope the following lemmas will be particularly helpful.

\begin{lemma}
\label{Lemma:AlphaLemma}
    For all $I \subseteq [5n+3]$ with $ \lvert I \rvert = 3$, there exist $p_I \in \mathbb{R}[\lambda]$ and $g_I$ in
    \begin{align*}
          \mathcal{G} = \Big\{  1, \,\lambda_{3k+1}+\lambda_{3k+2},\, \lambda_{3k+1}+\lambda_{3k+3} , \,  (\lambda_{3k+1}+\lambda_{3k+2})( \lambda_{3k^\prime+1}+\lambda_{3k^\prime+3}) \mid  k,k^\prime \in \{0,\dots,n-1\} \Big\}
    \end{align*}
     such that  $\alpha_{n,I} = \tfrac{p_I}{g_I}$, and $\sigma(p_I)  \subseteq \sigma(g_I)$.
\end{lemma}

\begin{proof}
      By Lemma~\ref{rem::zero-aJ}(iv) if $\alpha_{n,I} \neq 0$, then $[D_n^\top(\lambda)]_{[3],I}$ has at most two entries involving $\lambda$. These entries are of the following form:
    \begin{align}
        \label{Eq:AlphaLemma}
        \tfrac{-k\lambda_{3k+1}-(k+1)\lambda_{3k+2}}{\lambda_{3k+1}+ \lambda_{3k+2}} \quad \text{and} \quad  \tfrac{-k\lambda_{3k+1}+(1-k)\lambda_{3k+3}}{\lambda_{3k+1}+ \lambda_{3k+3}}.
    \end{align}
    Hence, $g_I$ can be chosen from $\mathcal{G}$. The inclusion  $\sigma(p_I)  \subseteq \sigma(g_I)$ follows from the observation that for all entries of $[D_n^{\top}(\l)]_{[3],I}$, the support of the numerator is contained in the support of the denominator.
\end{proof}

We recall from Proposition~\ref{prop::delta_Gw} that $\delta_n(\l)$ is given by 
\[ \delta_n(\lambda) = \prod_{k=0}^{n-1} (\lambda_{3k+1} + \lambda_{3k+2})(\lambda_{3k+1} + \lambda_{3k+3}) \lambda_{3k+1}^3. \]

\begin{lemma}
\label{Lemma_Gunaw_Vert_Deltan}
For every $n \in \mathbb{N}$, consider the vector 
	\begin{align}
 \label{Eq:DefOmegan}
	\nu_n := \begin{bmatrix} \nu  &  \ldots & \nu \end{bmatrix}^{\top} \in \mathbb{R}^{3n}, \quad \text{ where } \quad \nu := \begin{bmatrix} 5 & 0 & 0  \end{bmatrix}.
	\end{align}
Then $\nu_n$ is a vertex of $\N(\delta_n(\lambda))$. Furthermore, for $n \geq 2$ and $I \subseteq [5n-2]$, if $(\hat{z}_I,\nu_{n-1}) \in \sigma_-(q_{n-1})$, then $(z_I,\nu_{n}) \in \sigma_-(q_{n})$, where $\hat{z}_I \in \mathbb{R}^{5n-2}$ is obtained from $z_I \in \mathbb{R}^{5n}$ by removing the last $5$ coordinates. 
\end{lemma}

\begin{proof}
    In each monomial of $\delta_{n}(\lambda)$, the degree of the variable $\lambda_{3k+1}$ is at most $5$ for each $k = 0, \dots , n-1$. Moreover, there is a unique monomial in $\delta_{n}(\lambda)$ which is divisible by $\prod_{i=0}^{n-1}\lambda_{3k+1}^5$. This implies that $\nu_n$ is a vertex of $\N(\delta_n(\lambda))$. The second part follows from Proposition \ref{Lemma::Gunawardena_a0}.
\end{proof}

For $J \subseteq [5] \text{ and } |J| \leq 3$, let $a_J$ and $\hat{J}$ be as defined in \eqref{Eq::qnGunawardena_aJ} and in \eqref{Eq:DefofJhat}. Let $I \subseteq [5n+3], \, |I| = 3$ such that $\hat{J} \subseteq I$ and $\hat{J}^c \cap I = \emptyset$. 
If $\mu \in \sigma\big( \delta_n(\lambda) \alpha_{n,I}\big)$, then from the definition of $a_J$ we get that $(z_I,\mu)\in \sigma(a_J \prod_{i \in \hat{J}^c}h_i).$ Moreover, we have the following result.

\begin{lemma}
\label{Prop::SingsOf_fI}
     Let $J$ and $I$ be as defined above, and write $\alpha_{n,I} = \tfrac{p_I}{g_I}$. If $p_I$ has only positive (resp. negative) coefficients, then $(z_I,\mu) \in \sigma_+(a_J \prod_{i \in \hat{J}^c}h_i)$ (resp. $(z_I,\mu)  \in \sigma_-(a_J\prod_{i \in \hat{J}^c}h_i)$) for all $\mu \in \sigma\big( \delta_n(\lambda) \alpha_{n,I}\big)$.
\end{lemma}

\begin{proof}
    By Lemma \ref{Lemma:AlphaLemma}, $g_I$ divides $\delta_n(\lambda)$. Note that $\tfrac{\delta_n(\lambda)}{g_I}$ has only positive coefficients. Since $p_I$ has only positive (resp. negative)  coefficients, then 
    \begin{align}
    \label{Eq:ProofLemma:SingsOf_fI}
      \delta_n(\lambda) \alpha_{n,I} \prod_{i \in I^c} h_i = \tfrac{\delta_n(\lambda)}{g_I} p_I\prod_{i \in I^c} h_i
    \end{align}
 has only positive (resp. negative) coefficients. Thus, $(z_I,\mu) \in \sigma_+(a_J \prod_{i \in \hat{J}^c}h_i)$ (resp. $(z_I,\mu)  \in \sigma_-(a_J\prod_{i \in \hat{J}^c}h_i)$) for all $\mu \in \sigma\big( \delta_n(\lambda) \alpha_{n,I}\big)$. 
\end{proof}

The polynomial $q_n$ has $8n+3$ many variables, $h_1,\dots,h_{5n+3},\lambda_1,\dots,\lambda_{3n}$. Thus, $\sigma(q_n) \subseteq \mathbb{R}^{5n+3}\times\RR^{3n}$
Let  $\pr_1$ (resp. $\pr_2$) denote the coordinate projection from $\mathbb{R}^{5n+3}\times\RR^{3n}$ onto $\mathbb{R}^{5n+3}$ (resp. $\mathbb{R}^{3n}$). The following result finds vertices of the projected Newton polytope of $a_J$'s.

\begin{lemma}
\label{Prop::ProjOntoLambdas}
Let $\mu \in \Vertex\big(\N(\delta_n(\lambda))\big)$, and $\mathcal{I} \subseteq \{ I \subseteq [5n+3] \mid |I| = 3 \}$. Consider the polynomial
  \[f := \sum_{ I \in \mathcal{I}} \delta_n(\lambda)  \alpha_{n,I}  \prod_{i\in I^c}h_i.\]
  If there exists $I_0 \in \mathcal{I}$ such that $(z_{I_0},\mu) \in \sigma(f)$, then $\mu$ is a vertex of $\pr_2(\N(f))$.
\end{lemma}

\begin{proof} 
Let $\alpha_{n,I} = \tfrac{p_I}{g_I}$. By Lemma \ref{Lemma:AlphaLemma}, we have $\sigma(p_I) \subseteq \sigma(g_I)$. Since $g_I$ divides $\delta_n(\lambda)$, it follows that $\sigma(\delta_n(\lambda) \alpha_{n,I}) \subseteq \sigma(\delta_n(\lambda))$ for all $I \in \mathcal{I}$. This implies that
    \begin{align}
    \label{Eq:PrLambdas}
     \pr_2(\sigma(f)) \subseteq \bigcup_{I \in \mathcal{I}} \sigma( \delta_n(\lambda)\alpha_{n,I}) \subseteq \sigma(\delta_n(\lambda)).
    \end{align}

    Since $\mu$ is a vertex of $\N(\delta_n(\lambda))$, there exists $v \in \mathbb{R}^{3n}$ such that $v \cdot \omega < v \cdot \mu $ for all $\omega \in \N(\delta_n(\lambda)) \setminus \{ \mu \}$. From  $(z_{I_0},\mu) \in \sigma(f)$ follows that $\mu \in \pr_2(\sigma(f))$. Using \eqref{Eq:PrLambdas}, we conclude that $\mu$ is a vertex of $\pr_2(\N(f))$.
\end{proof}

In the following, we investigate the signed supports of different $a_J$'s.

\begin{proposition}
	\label{Prop::GunawardenaPositive}For every $n \in \mathbb{N}$,
	if $1 \in J$ or $2 \in J$, then $a_J$ has only non-negative coefficients. Furthermore, $a_J$ is not the zero polynomial for $J = \{ 2,3,5 \}$.
\end{proposition}

\begin{proof}
 If $1\in J$, write $I = \{r_1,r_2,5n-1\}$. For $\ell,\ell' \in \{0,\ldots,n-1\}$, by Lemma~\ref{rem::zero-aJ} $\alpha_{n,I}\neq 0$ in the following cases:
	\begin{itemize}
		\item[(i)] When $(r_1,r_2)=(1,2)$ and $(r_1,r_2)=(2,3)$, we get $\alpha_{n,I}=1$ and $\alpha_{n,I}=n$ respectively.
		\item[(ii)] When $(r_1,r_2)\in\{(1,5\ell+6),(2,5\ell+6),(2,5\ell+8),(3,5\ell+6),(5\ell+6,5\ell'+8)\}$, we have $\alpha_{n,I}=n-\ell-1.$
		\item[(iii)] When $(r_1,r_2)\in\{(1,5\ell+7),(2,5\ell+7),(3,5\ell+7),(5\ell+7,5\ell'+8)\}$, we obtain 
		$\alpha_{n,I}=(n-\ell -1)+\tfrac{ \lambda_{3\ell+3} }{ {\lambda_{3\ell+1}} + {\lambda_{3\ell+3}} } .$
	\end{itemize}
        Since the numerator of all these  $\alpha_{n,I}$ has only non-negative coefficients, from Lemma \ref{Prop::SingsOf_fI} it follows that $a_J$ with $1 \in J$ has only non-negative coefficients.
	
	If $2\in J$, $I = \{r_1,r_2,5n\}$. For $\ell,\ell' \in \{0,\ldots,n-1\}$, by Lemma~\ref{rem::zero-aJ} $\alpha_{n,I}\neq 0$ for following cases:
	\begin{itemize}
		\item[(i)] When $(r_1,r_2)=(1,2)$, we get $\alpha_{n,I}=1$.
		\item[(ii)] When $(r_1,r_2)=(2,3)$, we get $\alpha_{n,I}=n+\tfrac{  \lambda_{3n-1} }{ {\lambda_{3n-2}} + {\lambda_{3n-1}} } $.
		\item[(iii)] When $(r_1,r_2)\in\{(1,5\ell+6),(2,5\ell+6),(2,5\ell+8),(3,5\ell+6),(5\ell+6,5\ell'+8)\}$, we have $\alpha_{n,I}=(n-\ell-1)+\tfrac{  \lambda_{3n-1} }{ {\lambda_{3n-2}} + {\lambda_{3n-1}} }.$
		\item[(iv)] When $(r_1,r_2)\in\{(1,5\ell+7),(2,5\ell+7),(3,5\ell+7),(5\ell+7,5\ell'+8)\}$, we obtain 
		$\alpha_{n,I}=\left(n - \ell - 1\right)+\tfrac{\l_{3\ell+3}}{\l_{3\ell+1}+\l_{3\ell+3}}+\tfrac{\l_{3n-1}}{\l_{3n-2}+\l_{3n-1}}.$
	\end{itemize}
 
  Since the numerator of $\alpha_{n,I}$ in the above cases has only non-negative coefficients, using again Lemma \ref{Prop::SingsOf_fI}, we conclude that $a_J$ has only non-negative coefficients if $2 \in J$. 
	
	The second part of the proposition follows from  case (iii) listed above and Lemma~\ref{Prop::SingsOf_fI}.
\end{proof}

Proposition~\ref{Prop::GunawardenaPositive} implies the following corollary.

\begin{corollary}\label{Cor::positivepart-Gunawardena} The polynomials $a_{\{1\}},a_{\{2\}},a_{\{1,4\}},a_{\{1,5\}},a_{\{2,3\}},a_{\{2,4\}},a_{\{2,5\}},a_{\{1,4,5\}},a_{\{2,3,5\}},$ and $a_{\{2,4,5\}}$ have non-negative coefficients for all $n \in \mathbb{N}$. 
\end{corollary}

Next, in Corollary~\ref{Cor::p2property} we show that $q_n$ satisfies the closure property (P2) and to show connectivity it is enough to consider the following polynomial
\begin{small}
	\begin{align}\label{Eq::gunawardenapol_b}
	\begin{split}
	b(h,\l):= \Big(&a_\emptyset(\hat{h},\lambda) h_{5n+1}h_{5n+2} h_{5n+3} +a_{\{3\}}(\hat{h},\lambda) h_{5n+2}h_{5n+3} +a_{\{4\}}(\hat{h},\lambda) h_{5n+1}h_{5n+3} \\
	&+a_{\{5\}}(\hat{h},\lambda)h_{5n+1}h_{5n+2}+a_{\{3,5\}}(\hat{h},\lambda)h_{5n+2} +a_{\{4,5\}}(\hat{h},\lambda)h_{5n+1} \Big)h_{5n-1} h_{5n}.
	\end{split}
	\end{align}
\end{small}

\begin{corollary}\label{Cor::p2property} For $n \geq 2$, $q_n$ satisfies the closure property (P2). Additionally, if $b$ in \eqref{Eq::gunawardenapol_b}
	has one negative connected component, then $q_n$ also has exactly one negative connected component.

\end{corollary}

\begin{proof}
	To prove this result, we first claim that there exists a proper face of $\N(q_n)$ that contains all the negative exponent vectors of $q_n$. Since in all the terms of $q_n$ the variables $h_{5n-1},h_{5n}$ have exponents either 0 or 1, it follows that 
 \[(e_{5n-1} + e_{5n} )\cdot \omega \leq 2  \quad \text{ for all } \omega \in \sigma(q_n).\]
 From Proposition \ref{Lemma::Gunawardena_a0}, it follows that the above inequality is attained for some $\omega^\prime \in \sigma(q_n)$. Since the polynomial $a_{\{2,3,5\}}$ is non-zero by Proposition \ref{Prop::GunawardenaPositive}, there exists $\omega^{\prime \prime} \in \sigma(q_n)$ that is not contained in the face $\N_{e_{5n-1}+e_{5n}}(q_n)$. We conclude that  $\N_{e_{5n-1}+e_{5n}}(q_n)$ is a proper face of $\N(q_n)$.
 
 It is easy to check that $q_n$ restricted to $\N_{e_{5n-1}+e_{5n}}(q_n)$  is exactly the polynomial given by $b(h,\l)$. From Proposition \ref{Prop::GunawardenaZero} and Proposition \ref{Prop::GunawardenaPositive} it follows that each monomial of $q_n$ with negative coefficient is divisible by $h_{5n-1}h_{5n}$. Therefore, $\sigma_-(q_n) \subseteq \N_{e_{5n-1}+e_{5n}}(q_n)=\N(b)$. The corollary now follows from Theorem \ref{Thm::NegativeFace}.
\end{proof}

Next we will consider different summands of $b(h,\lambda)$.

\begin{proposition}
\label{Prop::Gunawardena_a35}For $n \geq 2$, there exists a strict separating hyperplane of the signed support for both $a_{\{ 3 \}} (\hat{h},\lambda)h_{5n-1}h_{5n}h_{5n+2}h_{5n+3}$ and $a_{\{ 3,5 \}}(\hat{h},\lambda)h_{5n-1}h_{5n}h_{5n+2}$. Furthermore, 
   \begin{itemize}
    \item[(a)]   $(z_{\{4,8,5n+1\}},\mu) \in \sigma_-( a_{\{ 3 \}} (\hat{h},\lambda)h_{5n-1}h_{5n}h_{5n+2}h_{5n+3})$ for all $\mu \in \sigma(\delta_n(\lambda))$.
    \item[(b)] $(z_{\{4,5n+1,5n+3\}},\mu) \in \sigma_-( a_{\{ 3,5 \}}(\hat{h},\lambda)h_{5n-1}h_{5n}h_{5n+2})$ for all $\mu \in \sigma(\delta_n(\lambda))$.
    \end{itemize}
\end{proposition}
\begin{proof}
    If $J=\{3\}$, let $I = \{r_1,r_2,5n+1\}$. By Lemma~\ref{rem::zero-aJ}, $\alpha_{n,I}\neq 0$ in following cases, where $\ell, \ell' \in \{0,\ldots,n-2\}$:
    \begin{itemize}
        \item[(i)] When $(r_1,r_2)=(1,2)$, we have $\alpha_{n,I}=1.$
        \item[(ii)] When $(r_1,r_2)=(1,3)$ and $(r_1,r_2)=(1,5\ell+8)$, we get $\alpha_{n,I}=-n+1$ and $\alpha_{n,I}=-n+\ell+2$ respectively.
        \item[(iii)] When $(r_1,r_2)\in\{(1,5\ell+4),(2,5\ell+4),(3,5\ell+4),(5\ell+4,5\ell'+8)\}$, we have $\alpha_{n,I}=-n+\ell+1.$
        \item[(iv)] When $(r_1,r_2)\in\{(1,5\ell+5),(2,5\ell+5),(3,5\ell+5),(5\ell+5,5\ell'+8)\}$, we obtain 
        $\alpha_{n,I}=\tfrac{ \left(-n + \ell + 1\right) \lambda_{3\ell+1} + \left(-n + \ell + 2\right) \lambda_{3\ell+2} }{ {\lambda_{3\ell+1}} + {\lambda_{3\ell+2}} }.$
    \end{itemize}
  From Lemma \ref{Prop::SingsOf_fI}, it follows 
    \begin{align*}
        \sigma_+( a_{\{ 3 \}} (\hat{h},\lambda)h_{5n-1}h_{5n}h_{5n+2}h_{5n+3}) = \{ (z_{\{1,2,5n+1\}},\mu) \mid \mu \in \sigma(\delta_n(\lambda)) \}.
    \end{align*}
    Since $(e_1+e_2) \cdot (z_I,\mu) \geq 0$ for all $I \subseteq [5n+3],\, |I| = 3$ and $\mu \in \sigma(\delta_n(\lambda)\alpha_{n,I})$, the hyperplane
    $\mathcal{H}_{e_1+e_2,0}$ is a strict separating hyperplane of the support of $a_{\{ 3 \}} (\hat{h},\lambda)h_{5n-1}h_{5n}h_{5n+2}h_{5n+3}$. 

    Let $J=\{3,5\}$, and let $I = \{r_1,5n+1,5n+3\}$. We only need to consider the following cases,
    \begin{itemize}
        \item[(i)] When $r_1=1$, we have $\alpha_{n,I}=1$.
        \item[(ii)] When $r_1=5\ell+4$, we have $\alpha_{n,I}=-n+\ell+1$.
        \item[(iii)] When $r_1=5\ell+5$, we obtain 
               $\alpha_{n,I}=\tfrac{ \left(-n + \ell + 1\right) \lambda_{3\ell+1} + \left(-n + \ell + 2\right) \lambda_{3\ell+2} }{ {\lambda_{3\ell+1}} + {\lambda_{3\ell+2}} }.$
    \end{itemize}
    The only non-negative numerator is given by $I=\{1,5n+1,5n+3\}$ and hence, $\mathcal{H}_{e_1,0}$ gives a strict separating hyperplane.

    To prove the second part of the proposition, we note that for $I=\{4,8,5n+1\}$ and $I=\{4,5n+1,5n+3\}$, $\alpha_{n,I}=-n+1<0$. Additionally, since $\alpha_{n,I}$ is an integer, $\sigma(\delta_n(\l)\alpha_{n,I})=\sigma(\delta_n(\l)).$ Hence, for $I=\{4,8,5n+1\}$ (resp. $I=\{4,5n+1,5n+3\}$) and for all $\mu\in \sigma(\delta_n(\l))$, $(z_I,\mu) \in \sigma_-(a_{\{ 3 \}} (\hat{h},\lambda)h_{5n-1}h_{5n}h_{5n+2}h_{5n+3})$ (resp. $\sigma_-(a_{\{ 3,5 \}}(\hat{h},\lambda)h_{5n-1}h_{5n}h_{5n+2})$).
\end{proof}

\begin{corollary}
\label{Cor::Gunawardena_b1}
For $n \geq 2$, the polynomial
   \[ b_0 := \Big(a_{\{3\}}(\hat{h},\lambda) h_{5n+2}h_{5n+3}+a_{\{3,5\}}(\hat{h},\lambda)h_{5n+2} \Big)h_{5n-1} h_{5n} \]
  has one negative connected component.
\end{corollary}

\begin{proof}
Let $b_{0,1} := a_{\{3\}}(\hat{h},\lambda) h_{5n+2}h_{5n+3}h_{5n-1} h_{5n}$ and $b_{0,2} := a_{\{3,5\}}(\hat{h},\lambda)h_{5n+2} h_{5n-1} h_{5n}$. The polynomial $b_{0,1}$ (resp. $b_{0,2}$) is the restriction of $b_0$ to the face of $\N_{e_{5n+3}}(b_0)$ (resp. $\N_{-e_{5n+3}}(b_0)$). 
Since $\sigma(b_{0,1})$ and $\sigma(b_{0,2})$ have strict separating hyperplanes (Proposition \ref{Prop::Gunawardena_a35}), $b_{0,1}$ and $b_{0,2}$ have one negative connected component by Proposition~\ref{Thm::SepHypThm}.

Let $\mu \in \sigma(\delta_n(\lambda))$  be a vertex of $\N\big(\delta_n(\lambda)\big)$.  By Proposition \ref{Prop::Gunawardena_a35}, $(z_{\{4,8,5n+1\}},\mu) \in \sigma_-(b_{0,1})$ and  $(z_{\{4,5n+1,5n+3\}},\mu) \in \sigma_-(b_{0,2})$. Moreover, by Proposition \ref{Prop::NegEdge}, $\Conv(z_{\{4,8,5n+1\}}, z_{\{4,5n+1,5n+3\}})$ is an edge of $\pr_1( \N(b_0 ) )$ and by Lemma~\ref{Prop::ProjOntoLambdas}, $\mu$ is a vertex of $\pr_2( \N(b_0) )$. Using Proposition \ref{Prop::EdgeLifting} we get that $\Conv( (z_{\{4,8,5n+1\}},\mu), (z_{\{4,5n+1,5n+3\}},\mu))$ is an edge of $\N(b_0)$. Since $b_{0,1}$ and $b_{0,2}$ have exactly one negative connected component and $\N(b_0)$ has an edge connecting two negative vertices of $\N(b_{0,1})$ and $\N(b_{0,2})$, we conclude by Theorem \ref{Thm::ParallelFaces} that $b_0$ has one negative connected component.
\end{proof}

\begin{proposition}
\label{Prop::Gunawardena_a45}
For $n \geq 2$, the supports $\sigma(a_{\{4\}} h_{5n-1} h_{5n}h_{5n+1}h_{5n+3})$, $
	\sigma(a_{\{5\}}h_{5n-1} h_{5n}h_{5n+1}h_{5n+2})$, and $\sigma(a_{\{4,5\}}h_{5n-1} h_{5n}h_{5n+1})$ have strict separating hyperplanes. For $n \geq 3$, we have
    \begin{itemize}
     \item[(a)] $(z_{\{1,4,5n+2\}},\mu) \in \sigma_-( a_{\{ 4 \}} (\hat{h},\lambda)h_{5n-1}h_{5n}h_{5n+1}h_{5n+3})$ for all $\mu \in \sigma(\delta_n(\lambda))$,
     \item[(b)] $(z_{\{2,4,5n+3\}},\mu) \in \sigma_-( a_{\{ 5 \}} (\hat{h},\lambda)h_{5n-1}h_{5n}h_{5n+1}h_{5n+2})$  for all $\mu \in \sigma(\delta_n(\lambda))$,
     \item[(c)] $(z_{\{4,5n+2,5n+3\}},\mu) \in \sigma_-( a_{\{ 4,5 \}} (\hat{h},\lambda)h_{5n-1}h_{5n}h_{5n+1})$  for all $\mu \in \sigma(\delta_n(\lambda))$.
      \end{itemize}
      Let $n = 2$ and $\nu_2$ be as given in \eqref{Eq:DefOmegan}. Then, 
          \begin{itemize}
     \item[(a')] $(z_{\{1,4,12\}},\nu_2) \in \sigma_-( a_{\{ 4 \}} (\hat{h},\lambda)h_{9}h_{10}h_{11}h_{13})$,
     \item[(b')] $(z_{\{2,4,13\}},\nu_2) \in \sigma_-( a_{\{ 5 \}} (\hat{h},\lambda)h_{9}h_{10}h_{11}h_{12})$,
     \item[(c')] $(z_{\{4,12,13\}},\nu_2) \in \sigma_-( a_{\{ 4,5 \}} (\hat{h},\lambda)h_{9}h_{10}h_{11})$.
      \end{itemize}
\end{proposition}

\begin{proof}
    Let $J=\{4\}$ and $I = \{r_1,r_2,5n+2\}$. By Lemma~\ref{rem::zero-aJ}, $\alpha_{n,I}\neq 0$ in the following cases, where $\ell,\ell' \in \{0,\ldots,n-2\}$:
    \begin{itemize}
        \item[(i)] When $(r_1,r_2)=(1,2)$ and $(r_1,r_2)=(1,3)$, we get $\alpha_{n,I}=1$ and $\alpha_{n,I} = \tfrac{  (-n+1)\lambda_{3n-2}+  (-n+2)\lambda_{3n} }{ \lambda_{3n-2} + {\lambda_{3n}} }$ respectively. 
        \item[(ii)] When $(r_1,r_2)\in\{(1,5\ell+4),(2,5\ell+4),(3,5\ell+4),(5\ell+4,5\ell'+8)\}$, we obtain the expression $\alpha_{n,I}= \tfrac{  (-n+\ell+1)\lambda_{3n-2}+  (-n+\ell+2)\lambda_{3n} }{ \lambda_{3n-2} + {\lambda_{3n}} }$
        \item[(iii)] When $(r_1,r_2)\in\{(1,5\ell+5),(2,5\ell+5),(3,5\ell+5),(5\ell+5,5\ell'+8)\}$, we obtain $\alpha_{n,I}= \tfrac{ \left(-n + \ell + 1\right) \lambda_{3\ell+1} \lambda_{3n-2} + \left(-n + \ell + 2\right) \lambda_{3\ell+2} \lambda_{3n-2} + \left(-n + \ell + 2\right) \lambda_{3\ell+1} \lambda_{3n} + \left(-n + \ell + 3\right) \lambda_{3\ell+2} \lambda_{3n} }{ {\left({\lambda_{3\ell+1}} + {\lambda_{3\ell+2}}\right)} {\left({\lambda_{3n-2}} + {\lambda_{3n}}\right)} } .$ 
        \item[(iv)] Finally, when $(r_1,r_2)=(1,5\ell+8)$, we get $\alpha_{n,I}=\tfrac{ \left(-n + \ell + 2\right) \lambda_{3n-2}+\left(-n + \ell + 3\right) {\lambda_{3n}} }{ {\lambda_{3n-2}} + {\lambda_{3n}} }.$ 
    \end{itemize}

For all $(z_I,\mu) \in \sigma( a_{\{4\}}h_{5n-1}h_{5n}h_{5n+1}h_{5n+3})$, we have
\begin{align}
\label{Eq::Proof::Gunawardena_a45}
 (-e_4) \cdot (z_I,\mu) \geq -1,
 \end{align}
with strict inequality if and only if $I = \left\{1,{4}, {5n+2}\right\}, \left\{{2}, {4},{5n+2}\right\}$, $\left\{{3}, {4}, {5n+2}\right\}, \left\{ {4},{5\ell+8}, {5n+2}\right\}$, $\ell \in \{0, \dots , n-2\}$. Using above cases and Lemma \ref{Prop::SingsOf_fI} if the inequality is strict, $(z_I,\mu)\in \sigma_-( a_{\{4\}}h_{5n-1}h_{5n}h_{5n+1}h_{5n+3})$ for all $\mu \in \sigma(\delta_n(\lambda)\alpha_{n,I})$. Hence, $\mathcal{H}_{-e_4,-1}$ is a strict separating hyperplane of $\sigma( a_{\{4\}}h_{5n-1}h_{5n}h_{5n+1}h_{5n+3})$. 

   Let $J=\{5\}$ and $I = \{r_1,r_2,5n+3\}$. By Lemma~\ref{rem::zero-aJ}, it is enough to compute $\alpha_{n,I}$ for $\ell,\ell' \in \{0,\ldots,n-2\}$ in the following cases:
    \begin{itemize}
        \item[(i)] When $(r_1,r_2)=(1,2)$, $(r_1,r_2)=(1,5\ell+6)$, $(r_1,r_2)=(2,5\ell+4)$, and $(r_1,r_2)=(5\ell+4,5\ell'+6)$ we get $\alpha_{n,I}$ as $1, n-\ell, -n+\ell+1,$ and $\ell-\ell'$, respectively.
         \item[(ii)] When $(r_1,r_2)=(1,5\ell+7)$ and $(r_1,r_2)=(2,5\ell+5)$, we get $\alpha_{n,I}=\tfrac{  \left(n - \ell\right)\lambda_{3\ell+1}+ \left(n - \ell+1\right)\lambda_{3\ell+3}}{\lambda_{3\ell+1} + {\lambda_{3\ell+3}}} $  and $\alpha_{n,I}=\tfrac{  \left(-n + \ell + 1\right)\lambda_{3\ell+1} + \left(-n + \ell + 1\right)\lambda_{3\ell+2} }{ {\lambda_{3\ell+1}} + {\lambda_{3\ell+2}} }$, respectively.
        \item[(iii)] When $(r_1,r_2)=(5\ell+4,5\ell'+7)$ and $(r_1,r_2)=(5\ell+5,5\ell'+6)$, we obtain the expression $\alpha_{n,I}=\left(\ell - \ell'\right)+\tfrac{ \lambda_{3\ell'+3} }{ {\lambda_{3\ell'+1}} + {\lambda_{3\ell'+3}} }$  and $\alpha_{n,I}=\left(\ell - \ell'\right)+\tfrac{ \lambda_{3\ell+2} }{ {\lambda_{3\ell+1}} + {\lambda_{3\ell+2}} }$, respectively.
        \item[(iv)] Finally, when $(r_1,r_2)=(5\ell+5,5\ell'+7)$, we get
        $\alpha_{n,I}=\left(\ell - \ell'\right)+\tfrac{\lambda_{3\ell+2}}{\lambda_{3\ell+1}+\lambda_{3\ell+2}}+\tfrac{   \lambda_{3\ell'+3} }{  {{\lambda_{3\ell'+1}} + {\lambda_{3\ell'+3}}} } .$
    \end{itemize}

     To show that $a_{\{5\}}$ has a separating hyperplane, we observe that:
    \begin{itemize}
    \item[(a)] if $1\in I$ and $2\notin I$ then  the numerator of $\alpha_{n,I}$ has only positive coefficients, and 
    \item[(b)] if $1\notin I$ and $2\in I$ then the numerator of $\alpha_{n,I}$ has only negative coefficients.
    \end{itemize}
Hence, $\mathcal{H}_{e_1 - e_2,0}$ is a strict separating hyperplane of $a_{\{ 5 \}} h_{5n-1}h_{5n}h_{5n+1}h_{5n+2}$ by Lemma  \ref{Prop::SingsOf_fI}. 
    
     Finally, let $J=\{4,5\}$ and let $I = \{r_1,5n+2,5n+3\}$. Then by Lemma~\ref{rem::zero-aJ}, $\alpha_{n,I}$ is non-zero in the following cases:
    \begin{itemize}
        \item[(i)] When $r_1=1$, we have $\alpha_{n,I}=\tfrac{ {\lambda_{3n-2}} + 2\lambda_{3n} }{ {\lambda_{3n-2}} + {\lambda_{3n}} } $.
        \item[(ii)] When $r_1=5\ell+4$, we have $\alpha_{n,I}=\tfrac{ \left(-n + \ell + 1\right)  \lambda_{3n-2}+ \left(-n + \ell + 2\right)\lambda_{3n} }{ {\lambda_{3n-2}} + {\lambda_{3n}} }$.
        \item[(iii)] When $r_1=5\ell+5$, we obtain $\alpha_{n,I}=\left(-n + \ell + 1\right)\tfrac{\l_{3\ell+2}}{{\lambda_{3\ell+1}} + {\lambda_{3\ell+2}}} +\tfrac{\l_{3n}}{{\lambda_{3n-2}} + {\lambda_{3n}}},$
    \end{itemize}
    where $\ell,\ell' \in \{ 0, \dots n-2\}$.
    
 To show that $\sigma( a_{\{4,5\}})$ has a strict separating hyperplane, we note that:
\begin{itemize}
    \item[(a)] if $1\in I$ and $4\notin I$ then the numerator of $\alpha_{n,I}$ has only positive coefficients. and 
    \item[(b)] if $1\notin I$ and $4\in I$ then the numerator of $\alpha_{n,I}$ has only negative coefficients. 
\end{itemize}
Hence, $\mathcal{H}_{e_1 - e_4,0}$ is a strict separating hyperplane of $\sigma( a_{\{ 4,5 \}} h_{5n-1}h_{5n}h_{5n+1})$ by Lemma~\ref{Prop::SingsOf_fI}.

We now focus on the second part of the proposition. We only explicitly prove part $(a)$, the cases (b) and (c) follow the similar argument and observations. The coefficients of the numerator of $\alpha_{n,I}$ are negative for $I=\{1,4,5n+2\}$ and $(z_{\{1,4,5n+2\}},\mu) \in \sigma( a_{\{ 4 \}} (\hat{h},\lambda)h_{5n-1}h_{5n}h_{5n+1}h_{5n+3})$. Moreover, the support of the numerator of $\alpha_{n,I}$ is same as the support of the denominator of $\alpha_{n,I}$ for $n \geq 3$. Hence, $\sigma( \delta_n(\lambda) ) = \sigma(\delta_n(\lambda)\alpha_{n,I})$. Now, the statement follows from Lemma \ref{Prop::SingsOf_fI}.

To prove (a'),(b'),(c'), we observe that 
\begin{align*}
     \alpha_{2,I} &= \tfrac{-\lambda_{4}}{\lambda_{4}+\lambda_{6}} ,\quad \text{ if } I = \{1,4,12\} \text{ or } I =  \{4,12,13\}, \qquad \text{and} \qquad \alpha_{2,I} &= -1 ,\qquad \text{ if } I = \{2,4,13\}.
\end{align*}
Thus, for $I \in \{  \{1,4,12\}, \{4,12,13\},\{2,4,13\} \}$, $\lambda_1^5\lambda_4^5$ is a monomial of $\delta_2(\lambda) \alpha_{2,I}$ with negative coefficients. This completes the proof.
\end{proof}

We now present the main result of this section.

\begin{theorem} \label{Thm::gunawardena_reduced_connected}
 For all $n \geq 2$, both for reduced weakly irreversible phosphorylation network \eqref{Eq::Gunawardena_reduced} and for the weakly irreversible phosphorylation network \eqref{Eq::Gunawardena}, the parameter region of multistationarity is non-empty and path connected.
\end{theorem}
\begin{proof}
   To prove that the parameter region of multistationarity is connected, it is enough to show that the polynomial $b$ in \eqref{Eq::gunawardenapol_b} has one negative connected component for all $n$ by Corollary~\ref{Cor::p2property}. We relabel the summands of $b$ as follows:
    \begin{equation}
    \label{Eq:DefOfb_Gunawardena}
     \begin{aligned}
     b_0 &= \big(a_{\{3\}}(\hat{h},\lambda) h_{5n+2}h_{5n+3} +a_{\{3,5\}}(\hat{h},\lambda)h_{5n+2}) \big)h_{5n-1} h_{5n}, \\
    b_1 &:=  a_\emptyset(\hat{h},\lambda) h_{5n-1} h_{5n}h_{5n+1}h_{5n+2} h_{5n+3}, \qquad  b_2 :=  a_{\{4\}}(\hat{h},\lambda)  h_{5n-1} h_{5n}h_{5n+1}h_{5n+3}, \\
    b_3 &:=    a_{\{5\}}(\hat{h},\lambda)h_{5n-1} h_{5n}h_{5n+1}h_{5n+2}, \quad \text{ and } \quad  b_4 := a_{\{4,5\}}(\hat{h},\lambda)h_{5n-1} h_{5n} h_{5n+1}.
    \end{aligned}
    \end{equation}

For all $n\geq 2$, we make the following observations:
\begin{enumerate}
    \item[(O1)] Polynomials $b_1$ and $b_2$ are the restrictions of the polynomial $b_1+b_2$ to the parallel faces $\N_{e_{5n+2}}(b_1+b_2)$ and $\N_{-e_{5n+2}}(b_1+b_2).$
    \item[(O2)]  Polynomials $b_3$ and $b_4$ are the restrictions of the polynomial $b_3+b_4$ to the parallel faces $\N_{e_{5n+2}}(b_3+b_4)$ and $\N_{-e_{5n+2}}(b_3+b_4)$. 
    \item[(O3)] Polynomials $b_1 + b_2$ and $b_3+b_4$ are the restrictions of $b_1+b_2+b_3+b_4$ to the parallel faces $\N_{e_{5n+3}}(b_1+b_2+b_3+b_4)$, and $\N_{-e_{5n+3}}(b_1+b_2+b_3+b_4).$
    \item[(O4)] Lastly, polynomials $b_0$ and $b_1+b_2+b_3+b_4$ are the restrictions of $b$ to the parallel faces $\N_{e_{5n+1}}(b)$, and $ \N_{-e_{5n+1}}(b)$.
\end{enumerate}

From Proposition \ref{Prop::Gunawardena_a45}, it follows that the support of $b_2, b_3,b_4$ have strict separating hyperplanes. Thus, $b_2,b_3,b_4$ have one negative connected component by Proposition~\ref{Thm::SepHypThm}. Moreover, the polynomial $b_0$ satisfies (P1) by Corollary~\ref{Cor::Gunawardena_b1}.

We will use these observations to first show that $b$ has one negative connected component for $n=2$ and then use induction on $n$.

For $n=1$, the critical polynomial $q_1$ only has positive coefficients. For $n=2$, by Lemma \ref{Lemma::Gunawardena_a0}, we have
\[b_1 =  q_{1}(\hat{h},\hat{\lambda}) (\lambda_{4} + \lambda_{5})(\lambda_{4} + \lambda_{6}) \lambda_{4}^3  h_{9} h_{10} h_{11}h_{12} h_{13}. \]
 Therefore, $\sigma_-(b_1+b_2) \subseteq \N_{-e_{12}}(b_1 + b_2)$. Since $b_2$ has one negative connected component, by observation (O1) above and Theorem \ref{Thm::NegativeFace}, $b_1 + b_2$ satisfies (P1). 

 Let $\nu_2=[5,0,0,5,0,0]^{\top}$. By Lemma \ref{Lemma_Gunaw_Vert_Deltan}, $\nu_2 \in \Vertex(\N(\delta_2(\lambda)))$. Let $I_0 := \{4,11,13\}$, $I_2 := \{1,4,12\}$, $I_3 = \{2,4,13\}$, $I_4 := \{4,12,13\}$. From Proposition~\ref{Prop::Gunawardena_a35} and Proposition~\ref{Prop::Gunawardena_a45} it follows that
\[ (z_{I_0},\nu_2 ) \in \sigma_-(b_0), \quad \quad (z_{I_2},\nu_2 ) \in \sigma_-(b_1+b_2), \quad (z_{I_3},\nu_2 ) \in \sigma_-(b_3), \quad (z_{I_4},\nu_2 ) \in \sigma_-(b_4).\]
 Using Lemma~\ref{Prop::ProjOntoLambdas}, Proposition \ref{Prop::NegEdge}, and Proposition \ref{Prop::EdgeLifting} we have the following edges joining negative vertices of parallel faces in (O2)-(O4): 
 \begin{itemize}
 \item[(i)] $\Conv( (z_{I_3},\nu_2 ) , (z_{I_4},\nu_2 ) )$ is an edge joining negative vertices of $\N(b_3+b_4)$,
 \item[(ii)] $\Conv( (z_{I_2},\nu_2 ) , (z_{I_4},\nu_2 ) )$ is an edge joining negative vertices of $\N(b_1 + b_2 +b_3+b_4)$, 
 \item[(iii)] $\Conv( (z_{I_0},\nu_2 ) , (z_{I_4},\nu_2 ) )$ is an edge joining negative vertices of $\N(b_0 +b_1 + b_2 +b_3+b_4)$.
 \end{itemize}
Since, $b_0,b_1+b_2,b_3,b_4$ satisfy (P1), the cases above and repeated application of Theorem~\ref{Thm::ParallelFaces} gives that $b$ has one negative connected component for $n = 2$.
   
    Let $n \geq 3$ and assume that $q_k$ has one negative connected component for all $2\leq k\leq n-1$. By the inductive assumption $q_{n-1}$ has one negative connected component and thus, by Proposition~\ref{Lemma::Gunawardena_a0} so does $b_1$.

    Let $I_0 := \{4,5n+1,5n+3\}$, $I_1 := \{1,4,5n-3\}$, $I_2 := \{1,4,5n+2\}$, $I_3 = \{2,4,5n+3\}$, $I_4 := \{4,5n+2,5n+3\}$ and let $\nu_n$ as defined in \eqref{Eq:DefOmegan}. Note that $\nu_n \in \Vertex(\N(\delta_n(\lambda) ))$ by Lemma \ref{Lemma_Gunaw_Vert_Deltan}. From Lemma~\ref{Lemma_Gunaw_Vert_Deltan}, Proposition~\ref{Prop::Gunawardena_a35} and Proposition~\ref{Prop::Gunawardena_a45} it follows that
\[ (z_{I_0},\nu_n) \in \sigma_-(b_0), \quad (z_{I_1},\nu_n) \in \sigma_-(b_1), \quad (z_{I_2},\nu_n) \in \sigma_-(b_2), \quad (z_{I_3},\nu_n) \in \sigma_-(b_3), \quad (z_{I_4},\nu_n) \in \sigma_-(b_4).\]
 Using Lemma~\ref{Prop::ProjOntoLambdas}, Proposition \ref{Prop::NegEdge}, and Proposition \ref{Prop::EdgeLifting} we conclude that: 
 \begin{itemize}
 \item[(i)] $\Conv( (z_{I_1},\nu_n) , (z_{I_2},\nu_n) )$ is an edge joining negative vertices of  $\N(b_1+b_2)$,
 \item[(ii)] $\Conv( (z_{I_3},\nu_n) , (z_{I_4},\nu_n) )$ is an edge joining negative vertices of $\N(b_3+b_4)$,
 \item[(iii)] $\Conv( (z_{I_2},\nu_n) , (z_{I_4},\nu_n) )$ is an edge joining negative vertices of $\N(b_1 + b_2 +b_3+b_4)$, 
 \item[(iv)] $\Conv( (z_{I_0},\nu_n) , (z_{I_4},\nu_n) )$ is an edge joining negative vertices of $\N(b_0 +b_1 + b_2 +b_3+b_4)$.
 \end{itemize}
 Using observations (O1)-(O4) and Theorem~\ref{Thm::ParallelFaces}, we conclude that $b$ has one negative connected component.
 
Since $q_n$ attains negative values, from \cite[Theorem 1]{PLOS_IdParaRegions} it follows that the multistationarity region is non-empty for the reduced network \eqref{Eq::Gunawardena_reduced}. From \cite[Theorem 5.1]{IntermRed}, it follows that  \eqref{Eq::Gunawardena} is multistationary for some choice of the parameters.
\end{proof}

\section{\bf Discussion}
{
In this work, we consider infinite families of strongly and weakly irreversible phosphorylation processes (\eqref{Eq::nsitenetwork} and \eqref{Eq::Gunawardena}). The networks in these families are indexed by the number of sites, $n$ available for phosphorylation. In particular, we are interested in the connectedness of the parameter region that enables multistationarity in these networks. We show that these regions are in fact connected for all~$n$. 

Our approach relies on studying the positivity of the critical polynomial associated with these networks using combinatorial methods. However, with increasing $n$ the size of the polynomial increases and makes the direct computation a challenging problem. We address this by exploiting Gale dual matrices to give a manageable representation of the polynomial and write the polynomial recursively over $n$ based on the structure of the network. This approach is not specific to these families of networks and a future research direction will study the critical polynomials for general networks using Gale dual matrices.

One of the key differences in the two families considered here is that the critical polynomial factorizes for the strongly irreversible phosphorylation networks. Equivalently, the Gale dual matrix of $N'\diag(E\l)A^{\top}$ for strongly irreversible phosphorylation networks is an integer matrix while the same for weakly irreversible networks is a parameterized matrix. This makes the analysis in Section~\ref{Section::ProofGunawardena} more complicated as compared to the one in Section~\ref{Section::Proofnsite}. Another avenue of research is to explore the networks for which the polynomial factorizes to give a simplistic description and how the analyses in Section~\ref{Section::ProofGunawardena} generalizes when it does not factor.
}

\medskip

\paragraph{\textbf{Acknowledgements. }}
The authors thank Elisenda Feliu for useful comments on the manuscript, which significantly improved its readability.
NK was supported by Independent Research Fund of Denmark.
MLT was supported by the European Union under the Grant Agreement no. 101044561, POSALG. Views and opinions expressed are those of the author(s) only and do not necessarily reflect those of the European Union or European Research Council (ERC). Neither the European Union nor ERC can be held responsible for them.

{\small

\bibliographystyle{plain}

}

\vspace{1cm}

\noindent
\footnotesize {\bf Authors' addresses:}

\smallskip

\noindent Nidhi Kaihnsa,
\  University of Copenhagen \hfill  {\tt nidhi@math.ku.dk}

\noindent Máté L. Telek,
\  University of Copenhagen \hfill  {\tt mlt@math.ku.dk}

\end{document}